\documentclass{article}
\usepackage[algo2e, linesnumbered,ruled,boxed,vlined,resetcount]{algorithm2e}
\usepackage{microtype}
\usepackage{graphicx}
\usepackage{subcaption}
\usepackage{booktabs} 

\usepackage{hyperref}
\usepackage{enumitem}
\usepackage{xspace}

\usepackage[preprint]{icml2026}

\usepackage[T1]{fontenc}   
\usepackage{inconsolata}   

\usepackage{xfrac}
\usepackage{amsmath}
\usepackage{amssymb}
\usepackage{mathtools}
\usepackage{amsthm}

\usepackage{xcolor}
\usepackage{placeins} 

\SetCommentSty{mycommfont}

\SetKwComment{Comment}{// }{}
\SetKw{Break}{break}
\SetKw{Do}{do}
\SetKw{KElse}{else}
\SetKw{Otherwise}{otherwise}
\SetKw{With}{with}
\DontPrintSemicolon 
\SetKwIF
  {WithProb}        
  {ElseWithProbIf}
  {ElseWithProb}
  {with probability}            
  {do}
  {else if}
  {else}
  {endwith}
\usepackage[capitalize,noabbrev]{cleveref}

\theoremstyle{plain}
\newtheorem{theorem}{Theorem}[section]

\newtheorem{lemma}[theorem]{Lemma}

\theoremstyle{definition}
\newtheorem{definition}[theorem]{Definition}

\theoremstyle{remark}
\newtheorem{remark}[theorem]{Remark}
\newtheorem{example}{Example}[section]

\usepackage[textsize=tiny]{todonotes}

\icmltitlerunning{Differentially Private Submodular Maximization with a Knapsack Constraint}

\input{src/macros}
\begin{document}

\twocolumn[
  \icmltitle{Differentially Private Submodular Maximization with a Knapsack Constraint}
  \icmlsetsymbol{equal}{*}

  \begin{icmlauthorlist}
    \icmlauthor{Ron Zadicario}{TAU}
    \icmlauthor{Tova Milo}{TAU}
  \end{icmlauthorlist}

  \icmlaffiliation{TAU}{Tel Aviv University, Israel}

  \icmlcorrespondingauthor{Ron Zadicario}{ronzadicario@mail.tau.ac.il}

  \icmlkeywords{Differential Privacy, Submodular Maximization, Knapsack}

  \vskip 0.3in
]

\printAffiliationsAndNotice{}  

\begin{abstract}
Submodular maximization subject to a knapsack constraint (SMK) is a fundamental problem in discrete optimization, with wide-ranging applications in machine learning and related fields. As these applications increasingly involve sensitive individual data, there is a growing need for high-utility algorithms that provide formal privacy guarantees.
In this work, we study the SMK problem under \emph{differential privacy}, considering both monotone and non-monotone objective functions. For monotone objectives, we propose a differentially private algorithm that achieves the optimal $(1-1/e)$-approximation ratio while significantly improving both additive error and query complexity over prior work.  We also present a more efficient algorithm for the same setting, achieving a $1/2$-approximation. For non-monotone objectives, we introduce, to our knowledge, the first differentially private algorithm with provable guarantees, achieving a $1/4$-approximation in expectation and an additive error comparable to the best known for monotone objective functions.
\end{abstract}




\section{Introduction}
\label{sec:intro}

Constrained maximization of submodular functions is a cornerstone of discrete optimization. Submodularity formalizes the notion of \emph{diminishing returns}: an element’s marginal contribution to a set decreases as the set grows. This property is ubiquitous in theory and practice, driving widespread interest in submodular maximization across machine learning (ML), data science, and related fields. As the problem is NP-hard in general, significant effort has been devoted to designing efficient approximation algorithms. These have been applied to a variety of tasks, including  feature selection~\cite{krause2005near}, personalized recommendation~\cite{mirzasoleiman2016fast}, influence maximization~\cite{kempe2003maximizing}, exemplar-based clustering~\cite{gomes2010budgeted}, information gathering~\cite{krause2008near}, active learning~\cite{esfandiari2021adaptivity}, data subset selection~\cite{wei2015submodularity}, and interpretable rule learning~\cite{lakkaraju2016interpretable}. Many of these applications involve non-monotone objectives due to diversification or regularization penalties.

One of the earliest and most studied problem variants involves elements with associated costs, and feasible solutions are constrained by a budget. This constitutes the classical problem of submodular maximization subject to a knapsack constraint (SMK), which has attracted sustained attention for over four decades~\cite{wolsey1982maximising}.
 More recently, submodular maximization has been applied in settings involving sensitive individual data, where it is crucial that the optimization procedure simultaneously ensures rigorous privacy guarantees alongside strong utility. As an illustration, consider the following feature selection problem
\begin{example}\label{example}
A sensitive medical dataset
$\{(\mathbf{x}_i, y_i)\}_{i=1}^n$
associates each individual $i$ with a feature vector
$\mathbf{x}_i = (\mathbf{x}_i(1), \ldots, \mathbf{x}_i(m))$
and a binary label $y_i$ indicating the presence of a disease.
Each feature $j \in [m]$ (e.g., blood test result) incurs an acquisition cost $c_j$.
The goal is to select a subset of features $S \subseteq [m]$
that enables accurate prediction of $y$
while satisfying the budget constraint
$\sum_{j \in S} c_j \le B$.
Following \citet{krause2005near,mitrovic2017differentially},
a natural approach models this task as maximizing a submodular objective
that captures the mutual information between the selected features and the label.
Crucially, the feature selection process must preserve the privacy
of all individuals contributing to the dataset.
\end{example}
   
In this work, we study the monotone and non-monotone SMK problem under \emph{differential privacy}~\cite{dwork2006calibrating}, a rigorous notion that has become the standard for statistical analysis of sensitive data under strong privacy guarantees. Intuitively, the output distribution of a differentially private (DP) algorithm changes only slightly when a single individual’s data is modified, thereby ensuring that individual-level information is not disclosed.
DP submodular maximization has attracted increasing attention, most commonly in settings where the ground set $\mathcal N$ is public, while the submodular objective $f_D : 2^{\mathcal N} \to \mathbb{R}$ depends on a sensitive dataset $D$ in a problem-specific manner. The goal is to privately find a subset $S\subseteq \cN$ that approximately maximizes $f_D$ subject to given feasibility constraints. Most prior work has focused on cardinality constraints~\cite{gupta2010differentially,mitrovic2017differentially} or matroid constraints~\cite{rafiey2020fast,chaturvedi2021differentially}, while knapsack constraints have, to our knowledge, only been addressed by \citet{sadeghi2021differentially}.
 However, their work leaves room for improvement in several important aspects. 

First, their approach relies on exact evaluation of the \emph{multilinear extension} of the submodular function, 
which may require as many as $2^{|\cN|}$ value oracle queries to $f_D$\footnote{While the multilinear extension can be approximated via sampling in polynomial time, doing so would require a separate analysis beyond that in~\cite{sadeghi2021differentially}. Even with such approximations, combinatorial algorithms are typically faster than continuous relaxation-based approaches~\cite{feldman2023practical}.}.
Second, it is restricted to monotone functions, thereby excluding a variety of important applications. Third, the additive error in their utility guarantee exhibits a suboptimal dependence on both the ground set size and the element costs (see \Cref{tab:results_comparison} for a comparison with our results).

Motivated by these limitations, we propose DP algorithms for monotone and non-monotone SMK that bridge existing gaps in utility and query complexity. The privatization of combinatorial algorithms for SMK gives rise to challenges absent in matroid or cardinality constraint settings, requiring different techniques.
 These challenges, along with the formal statements of our results, are detailed in \Cref{section:challenges_results}. Our  contributions are summarized as follows:

\paratitle{Monotone}  
We present a DP algorithm for monotone SMK that  
\textit{(i)} achieves the optimal $(1-1/e)$-approximation,  
\textit{(ii)} has a query complexity comparable to that of the currently most practical non-private $(1-1/e)$-approximation algorithm, and in particular polynomial for any submodular objective,
\textit{(iii)} achieves an additive error that improves upon prior work in two ways: it has polylogarithmic dependence on the ground set size, rather than polynomial, and scales with the maximum size of a feasible solution, rather than the minimum element cost. Row (i) of \Cref{tab:results_comparison} summarizes this result.  
We further introduce a faster algorithm achieving $1/2$-approximation with improved additive error,
as shown in row (ii) of \Cref{tab:results_comparison}.

\paratitle{Non-monotone}  
We present the first DP algorithm for non-monotone SMK with provable guarantees. Our algorithm achieves a $1/4$-approximation in expectation, matching the state-of-the-art among combinatorial non-private algorithms. Furthermore, it incurs an additive error comparable to that of the monotone case while maintaining a practical query complexity. This result is summarized in row (iii) of \Cref{tab:results_comparison}.

\begin{table*}
\caption{Guarantees of $(\eps,\delta)$-DP algorithms for SMK with a $\Delta$-sensitive objective function. Constants are omitted. Here~$n$ denotes the size of the ground set, $k$ denotes the maximum cardinality of a feasible solution,  $L=\max_u f({u})$, and  $\beta$ is the failure probability. The budget is normalized to $B=1$, which implies $1/c_{\min}\ge k$.  \textsuperscript{$\dagger$} This bound is derived from Lemma 4 in \cite{sadeghi2021differentially}. \textsuperscript{$\ddagger$}  Their analysis relies on exact evaluation of the multilinear extension, which, in the value-oracle model, requires $2^n$ queries in general. }
\centering
{\footnotesize
\renewcommand{\arraystretch}{1.8}
\setlength{\tabcolsep}{6pt}
\begin{tabular}{@{}l l c l c l@{}}
\toprule
\textbf{Monotone} & \textbf{Approximation} & \textbf{Additive Error} & \textbf{ Query Complexity} & \textbf{Algorithm} & \\
\midrule
Yes & $1-1/e$ & $\tfrac{\Delta k^{1.5}}{\eps}
\sqrt{\log\tfrac{n}{\delta\beta}}
 \log\tfrac{n}{\beta}
$ & $O(\beta^{-1}n^3k)$ & \cref{alg:two_guess_greedy} {\scriptsize  \textbf{(Ours)}} & (i) \\

Yes & $1/2$ & $\tfrac{\Delta k^{1.5}}{\varepsilon} \sqrt{\log \tfrac{1}{\delta}} \,\log \tfrac{n}{\beta}$ & $O(nk)$ & \Cref{algorithm:greedyplus} {\scriptsize  \textbf{(Ours)}} & (ii)\\

No & $1/4~$ {\scriptsize (Expected)} & $\tfrac{\Delta k}{\varepsilon} \sqrt{(k + \log \tfrac{1}{\delta}) \log \tfrac{1}{\delta}} \,\log \tfrac{n}{\beta}$ & $O(nk)$ & \Cref{algorithm:subsample_greedy} {\scriptsize  \textbf{(Ours)}} & (iii) \\

\midrule
Yes & $1-1/e$ & $\tfrac{\Delta}{c_{\min}\eps}\sqrt{\tfrac{1}{\beta}\log\tfrac{1}{\delta}} \,n\log\tfrac{n}{\beta}+\tfrac{\beta L }{c_{\min}^2}$\textsuperscript{$\dagger$}
 & ${O(2^n)}^\ddagger$ & \cite{sadeghi2021differentially}  & (iv)  \\
\bottomrule
\end{tabular}
}
\label{tab:results_comparison}
\end{table*}

\section{Related Work}
\label{sec:related}
In this section, we survey related work on non-private SMK, DP submodular maximization, and lower bounds. Throughout, $k$ denotes the maximum size of a feasible solution, and $n$ the ground set size.

\paratitle{Non-Private Monotone SMK}
For cardinality-constrained submodular maximization, the classical greedy algorithm achieves a tight $(1-1/e)$-approximation \cite{nemhauser1978analysis}. Its SMK extension, Density-Greedy \cite{wolsey1982maximising}, selects elements of largest \emph{density}, i.e., the ratio between marginal contribution and cost. While it lacks a constant-factor approximation~\cite{feldman2009private}, several variants obtain provable but suboptimal guarantees~\cite{wolsey1982maximising, khuller1999budgeted}. \citet{sviridenko2004note} first achieved the optimal~$(1-1/e)$-approximation by guessing three elements from the optimal solution and running Density-Greedy on the residual instance. This was later improved to two elements by \citet{feldman2023practical}. The resulting \texttt{Two-Guess-Greedy} algorithm remains the most practical tight-approximation method, yet still requires enumerating all $O(n^2)$ pairs. Recent works have obtained a $(1/2-\varepsilon)$-approximation with near-linear query complexity~\cite{yaroslavtsev2020bring, han2021approximation}, while \citet{li2022submodular} achieved the same guarantee using $O_{\varepsilon}(n)$ queries.

\paratitle{Non-Private Non-Monotone SMK}
The best-known approximation for non-monotone SMK is $0.385$, obtained by combining the continuous optimization framework of \citet{buchbinder2019constrained} with the rounding scheme of \citet{kulik2013approximations}. This approach requires computationally expensive operations such as optimizing the multilinear extension, which motivates faster combinatorial algorithms. \citet{mirzasoleiman2016fast} obtained a deterministic $1/10$-approximation, followed by a deterministic $1/6$-approximation by \citet{han2021approximation}. \citet{amanatidis2020fast} introduced a randomized density-greedy variant where each selected element is discarded with some probability, achieving a $1/(3+2\sqrt{2})$-approximation in expectation. This was later improved by \citet{han2021approximation} to $1/4$, which remains the best practical guarantee.

\paratitle{Differentially Private Submodular Maximization}  
DP submodular maximization has been studied in two main settings. The first line of work focuses on \emph{$\Delta$-decomposable} objective functions~\cite{gupta2010differentially, chaturvedi2021differentially,chaturvedi2023streaming,ghazi2024individualized}. These take the form $f(S)=\sum_x f_x(S)$, where each individual $x$ contributes a submodular function $f_x:2^\cN\to[0,\Delta]$.
However, many applications involve submodular functions with more general dependence on the sensitive dataset (e.g., the objective in \Cref{example}). A second line of work therefore considers the more general setting of functions with \emph{bounded sensitivity} $\Delta$, which we adopt in this work.

As a simple baseline, applying the exponential mechanism~\cite{mcsherry2007mechanism} to select among all feasible subsets yields a $1$-approximation with additive error $O(\Delta k \log n / \eps)$. However, it incurs an exponential query complexity.
\citet{mitrovic2017differentially} proposed efficient algorithms for monotone objectives under cardinality constraints, achieving a $(1-1/e)$-approximation with $\Ot(\Delta k^{1.5}/\eps)$ additive error.\footnote{$\Ot$ hides polylogarithmic factors in the problem parameters.} For non-monotone objectives, they obtained a  $(1-1/e)/e$-approximation in expectation with the same error, and further extended their results to $p$-extendable system constraints. Other notable works include \citet{rafiey2020fast} on matroid constraints, \citet{sadeghi2021differentially} on bounded-curvature objectives, and \citet{chaturvedi2023streaming} on streaming algorithms. In contrast, SMK has been addressed only by the aforementioned work of \citet{sadeghi2021differentially}, who among other results, gave an $(\eps,\delta)$-DP algorithm for monotone functions achieving the guarantees stated in row (iv) of \Cref{tab:results_comparison}.

\paratitle{Lower Bounds}
Existing lower bounds on the expected additive error for DP submodular maximization under a cardinality constraint $k$ are $\Omega(k \log(n/k)/\eps)$ for $\eps$-DP algorithms~\cite{gupta2010differentially}, and $\Omega(k c \log(\eps/\delta)/\eps)$ for $(\eps,\delta)$-DP, $c$-approximation algorithms (assuming $n \ge k (e^\eps - 1)/\delta$ and $c \ge 4\delta/(e^\eps - 1)$)~\cite{chaturvedi2023streaming}. As cardinality is a special case of a knapsack constraint, these bounds extend to our setting. While our additive error exceeds them by $\tilde{O}(\sqrt{k})$, these lower bounds are derived from $1$-decomposable objectives.
In contrast, we consider the broader class of bounded-sensitivity submodular functions, and the dependence on $k$ in our bounds matches the best known for polynomial-time algorithms in the simpler cardinality-constrained setting.

\section{Preliminaries}
\label{sec:prelim}
In this section, we introduce notation, submodular functions, and key differential privacy concepts used throughout the paper. Additional details on DP generalized selection and omitted pseudo-codes are provided in \Cref{appendix:private_selection}.

\paratitle{Differential Privacy} 
A dataset is a tuple $D=(x_1,\dots,x_m)\in \cX^m$ of items from domain $\cX$ where $x_i$ is individual $i$'s data. Two datasets $D$ and $D'$ are called \emph{neighboring} ($D\sim D'$) if they differ in at most one entry. Differential privacy ensures that the output distribution of an algorithm changes minimally when one record is modified:
\begin{definition}[Differential Privacy,~\cite{dwork2006differential}]
A randomized algorithm $\cA$ is $(\eps,\delta)$-differentially private (DP) if for any $D \sim D'$ and $S \subseteq \mathrm{Range}(\cA)$,
\[
\Pr[\cA(D)\in S] \le e^\eps \Pr[\cA(D')\in S] + \delta.
\]
If $\delta=0$, we say $\cA$ is $\eps$-DP.
\end{definition}

Our privacy analysis uses the following standard composition results. 
\begin{theorem}[\citet{dwork2010boosting}]
\label{thm:composition}
Let $\cA_1,\dots,\cA_k$ be $\eps_0$-DP algorithms. Their $k$-fold adaptive composition $\cA_{[k]}$, which outputs $y_i=\cA_i(D, y_1,\dots,y_{i-1})$ for $i=1,\dots,k$, is: 
\textit{(i)}  $k\eps_0$-DP (Basic Composition), and \textit{(ii)}$(\sqrt{2k\log(1/\delta)}\eps_0 + k\eps_0(e^{\eps_0}-1),\, \delta)$-DP for any $\delta>0$ (Advanced composition).
To achieve $(\eps, \delta)$-DP for $\eps\le 1$, it suffices that each $\cA_i$ is $\frac{\eps}{2\sqrt{2k\log(1/\delta)}}$-DP. 
\end{theorem}

The \emph{sensitivity} of a function $q:\cX^m\to \R$ measures the maximum change in its value across neighboring datasets:
\[
\Delta_q = \max_{D\sim D'} |q(D)-q(D')|.
\]
A canonical DP mechanism for privately answering numerical queries is the Laplace mechanism:
\begin{theorem}[Laplace mechanism~\cite{dwork2014algorithmic}]\label{thm:laplace_guarantees}
Suppose $q:\cX^m \to \R$ has sensitivity $\Delta_q$. The mechanism returning
$q(D) + Y$ with $Y\sim \mathrm{Lap}(\Delta_q/\eps)$ is $\eps$-DP.
\end{theorem}

\paratitle{Private Selection}  
Selecting an element $u$ from a finite set $\cC$ to approximately maximize a quality score $q(u;D)$ is a fundamental task in private data analysis.
The \emph{Report Noisy Max} mechanism with exponential noise\footnote{The PDF of $\mathrm{Exp}(\lambda)$ is $f(x; \lambda) = \lambda e^{-\lambda x}$ for $x \ge 0$.}~\cite{dwork2014algorithmic}, denoted $\RNM{\eps}{\cC,q}{D}$, is a standard approach for solving this task accurately and efficiently.
This mechanism outputs $u^* \in \arg\max_{u\in\cC} \{q(u;D) + Y_u\}$, where $Y_u \sim \mathrm{Exp}(\frac{\eps}{2\Delta})$ and $\Delta = \max_{u\in\cC} \Delta_{q(u;\cdot)}$ is the maximal sensitivity of $q$ across all candidates.

\begin{theorem}[\citealp{dwork2014algorithmic}]\label{thm:RNM_main}
$\RNM{\eps}{\cC,q}{D}$ is $\eps$-DP. Moreover, with probability $1-\beta$, the output $u^*$ satisfies $q(u^*;D) \ge \max_{u\in\cC} q(u;D) - (2\Delta/\eps)\log(|\cC|/\beta)$.
\end{theorem}

The error of \RNMT scales with the maximal candidate sensitivity, which is insufficient in our setting (see \Cref{section:challenges_results}). To achieve improved utility, we leverage the generalized mechanism of~\citet{raskhodnikova2015efficient}, denoted $\GenRNM{\eps}{\cC,q,\beta}{D}$. Given a failure probability $\beta$, it applies \RNMT to carefully constructed auxiliary scores that reduce sensitivity while appropriately penalizing high-sensitivity candidates. Its formal definition is given in \Cref{appendix:private_selection}.

\begin{theorem}[\citealp{raskhodnikova2015efficient}]\label{thm:genRNM_guarantee}
$\GenRNM{\eps}{\cC,q,\beta}{D}$ is $\eps$-DP. Moreover, With probability $1-\beta$, the output $u^*$ satisfies $q(u^*;D) \ge \max_{u\in\cC} \{q(u;D) - (2\Delta_{q(u;\cdot)}/\eps)\log(|\cC|/\beta)\}$.
\end{theorem}
We further utilize the generalized selection framework of \citet{cohen2023generalized} to select the best output from a  collection of private subroutines. This framework provides tighter privacy guarantees than the composition theorems, as the privacy loss is independent of the total number of subroutines. See \Cref{appendix:private_selection} for details.

\paratitle{Submodular Functions}
Let $\cN$ be a ground set of size $n$. For $f:2^\cN \to \R$, the marginal contribution of $u$ to $S$ is $f(u\mid S) \eqdef f(S \cup \{u\}) - f(S)$. $f$ is \emph{monotone} if $f(S)\le f(T)$ for all $S\subseteq T$, \emph{non-negative} if $f(S)\ge 0$ for all $S$, and \emph{submodular} if $f(u\mid S) \ge f(u\mid T)$ for all $S\subseteq T$ and $u\notin T$. In the DP setting, $f_D$ depends on sensitive dataset $D$, and its  maximal sensitivity is given by
\[
\Delta_f = \max_{S\subseteq\cN, D\sim D'} |f_D(S)-f_{D'}(S)|.
\]
For brevity, we omit the dataset subscript $D$ where it is clear from context.

\paratitle{Submodular Maximization Subject to a Knapsack Constraint (SMK)}
Given a submodular $f:2^{\cN}\to\R$, a positive cost function $c:\cN \to \R_{>0}$, and budget $B>0$, SMK aims to maximize $f(S)$ subject to the feasibility constraint $c(S) \eqdef \sum_{v\in S} c(v) \le B$. We let $k$ denote the maximal cardinality of a feasible set and $c_{\min}\eqdef \min_{v\in\cN} c(v)$. The \emph{density} of $u\in\cN$ w.r.t. $S\subseteq \cN\setminus\{u\}$ is given by $f(u\mid S)/c(u)$. Following standard practice, we assume $c(v) \le B$ for all $v \in \cN$ and $B = 1$, as elements exceeding the budget are never feasible and our algorithms are invariant under scaling of the costs.
We assume access to $f$ via a \emph{value oracle} that returns $f(S)$ for a given $S$. The total number of oracle queries performed by an algorithm serves as a proxy for its time complexity.

\section{Challenges and Results}\label{section:challenges_results} 
In this section, we outline the technical challenges of DP SMK, describe our proposed resolutions, and state our main results. 

\subsection{Monotone Objective Functions}\label{section:challenges_results_monotone}
To obtain a combinatorial algorithm with improved utility, we build on the \texttt{Two-Guess-Greedy} algorithm of \cite{feldman2023practical}. Yet, its privatization is more involved than that of the standard greedy algorithm for cardinality constraints. 

\paratitle{Scores with Varying Sensitivity}
In the cardinality-constrained setting, differential privacy is typically achieved by replacing the greedy selection step with the Exponential Mechanism \cite{mcsherry2007mechanism}, a technique that yields strong guarantees in that context \cite{gupta2010differentially, mitrovic2017differentially}. However, this direct substitution is suboptimal for knapsack constraints, as density scores exhibit varying sensitivities that depend on element costs.
A naive application of standard selection mechanisms would incur an additive error proportional to $\Delta_f \cdot B/c_{\min}$, 
a uniform sensitivity bound for all density scores. Since $B/c_{\min} \ge k$, this analysis would introduce a higher-order dependence on $k$, the maximal size of a feasible solution.

To ensure the additive error is independent of the worst-case sensitivity of a density score, we leverage the selection mechanism of \citet{raskhodnikova2015efficient}, which provides fine-grained utility bounds with respect to individual candidate sensitivities. Using this guarantee, we show that our algorithm attains a smaller additive error proportional to $\Delta_f$. This dependence matches the cardinality-constrained setting and avoids looser worst-case bounds, despite the higher sensitivity of density scores.

\paratitle{Composition Bottleneck} 
Another challenge stems from the $O(n^2)$ iterations required by \texttt{Two-Guess-Greedy}, which pose a formidable composition bottleneck. A standard analysis based on advanced composition \cite{dwork2014algorithmic} inevitably yields an error that scales polynomially with $n$. Hence, bridging the gap to the polylogarithmic dependence achieved in state-of-the-art results for matroid and cardinality constraints requires a different approach. 

We observe that the highly parallelizable structure of \texttt{Two-Guess-Greedy} is advantageous not only for efficiency, but also for privacy. Specifically, the algorithm aligns naturally with a selection framework where the best output is chosen from a collection of DP subroutines. In our case, these subroutines are privatized adaptations of the density-greedy algorithm. Based on the generalized selection framework of \citet{cohen2023generalized}, we randomize the number of times each enumeration step is invoked. This allows us to bypass the standard composition penalty and  obtain an additive error that is polylogarithmic in $n$, a significant improvement over prior work.

\paratitle{Main Results}
Our approach is formalized in \Cref{alg:two_guess_greedy}, for which we establish the guarantee stated in \Cref{thm:two_guess_greedy_intro}. We provide the algorithm and an overview of the analysis in \Cref{sec:monotone}, with the full proofs deferred to \Cref{appendix:monotone}.
\begin{theorem}\label{thm:two_guess_greedy_intro}
Let $f_D:2^\cN\to\R_{+}$ be a monotone submodular function with sensitivity $\Delta$.  
Then, for any $\eps,\beta\in(0,1)$ and $\delta>0$, there exists an $(\eps,\delta)$-DP algorithm for SMK that returns a feasible set $S\subseteq\cN$ such that, with probability at least $1-\beta$,
\[
f(S)\ge (1-\tfrac{1}{e})\cdot f(\OPT)
-O\Big(
\tfrac{\Delta k^{1.5}}{\eps}
\sqrt{\log\tfrac{n}{\delta\beta}}
 \log\tfrac{n}{\beta}
\Big).
\]
Moreover, the algorithm makes $O(\beta^{-1}n^3k)$ oracle queries.
\end{theorem}
As we show in \Cref{appendix:monotone}, the query complexity in \Cref{thm:two_guess_greedy_intro} can be further improved to $O(n^3k\log\beta^{-1})$ at the cost of an additional $\log(\beta^{-1})$ factor in the error term.

While \texttt{Two-Guess-Greedy} is, to our knowledge, the most practical existing approach for achieving a tight $(1-1/e)$ approximation, its computational cost remains high for large-scale instances. Thus, we additionally introduce a more efficient algorithm based on a single execution of a density-greedy variant~\cite{yaroslavtsev2020bring}. 
This result is stated in \Cref{thm:greedyplus_intro}, with details provided in \Cref{appendix:monotone}.
\begin{theorem}\label{thm:greedyplus_intro}
Let $f_D:2^\cN\to\R_{+}$ be a monotone submodular function with sensitivity $\Delta$.  
Then, for any $\beta\in(0,1)$ and $\eps,\delta>0$, there exists an $(\eps,\delta)$-DP algorithm for SMK that returns a feasible set $S\subseteq\cN$ such that, with probability at least $1-\beta$,
\[
f(S) \ge \tfrac{1}{2} f(\mathrm{OPT}) - O\Big( \tfrac{\Delta k^{1.5}}{\eps} \sqrt{\log \tfrac{1}{\delta}} \,\log \tfrac{n}{\beta} \Big)
\]
Moreover, the algorithm makes $O(nk)$ oracle queries.
\end{theorem}

\subsection{Non-Monotone Objective Functions}\label{section:challenges_results_nonmonotone}
At present, the most practical non-private approaches for non-monotone SMK combine density-greedy selection with randomized sampling by which selected elements are discarded with some fixed probability~\cite{han2021approximation,amanatidis2020fast}, ensuring high utility in expectation. Our proposed algorithm follows this paradigm, but its privatization introduces three primary obstacles.

\paratitle{Composition Bottleneck} 
The uncertainty in the number of selection steps creates a challenge for maintaining a tight privacy analysis. Since selected elements can be discarded, the algorithm may perform up to $n$ selections. A standard application of advanced composition would yield an error term scaling with $\sqrt{n}$, rather than the desired logarithmic dependence. Moreover, unlike the non-private setting, initial subsampling of the ground set is not equivalent to post-selection discarding, since the output distribution of a DP selection mechanism depends on its entire candidate set.

Instead, we derive a concentration bound for the number of private selections, showing that the number of required private selections is stochastically dominated by a sum of independent geometric random variables that concentrates around $O(k)$. This ensures that, with high probability, the algorithm does not exceed the composition privacy budget.

\paratitle{Robustness to Negative Marginal Gains}
Beyond composition, the non-monotone setting presents a challenge for private greedy selection. Analyses for non-monotone SMK, such as that of \citet{han2021approximation}, rely on selecting elements with strictly positive marginal gain to ensure progress. In the DP setting, however, additive noise can cause elements with negative marginal contributions to be selected.

Instead of halting when no unselected elements with positive marginal gain exist, we continue adding feasible elements. Then, selecting the approximately best observed solution using \RNMT allows us to focus the analysis on early iterations, where feasible elements with sufficiently large gains exist. 
By extending the analysis of \citet{han2021approximation} to accommodate a \emph{bounded additive slack}, and showing that elements with sufficiently large contribution are selected with high probability in early iterations, we obtain a $1/4$-approximation even when noise masks the sign of an element’s marginal contribution.

\paratitle{Scores with Varying Sensitivity}
Finally, the challenge of varying sensitivities identified in the monotone case (\Cref{section:challenges_results_monotone}) persists here, as the density scores remain dependent on element costs. We resolve this in the same manner, leveraging \GenRNMT  to ensure that the additive error does not depend on the worst-case sensitivity of a density score. 

\paratitle{Main Result}
Our approach is formalized in \Cref{algorithm:subsample_greedy}, for which we establish the guarantee in \Cref{thm:samplegreedy}. We present the algorithm and analysis in \Cref{sec:non-monotone}, with full proofs deferred to \Cref{sec:appendix_nonmonotone}.

\begin{theorem}\label{thm:samplegreedy}
Let $f_D:2^\cN\to\R_{+}$ be a submodular function with sensitivity $\Delta$.
Then, for any $\eps,\beta\in(0,1)$ and $\delta>0$, there exists an $(\eps,\delta)$-DP algorithm for SMK that returns a feasible set $S\subseteq\cN$ such that, with probability at least $1-\beta$,
\[
\mathbb{E}[f(S)] \ge \tfrac{1}{4} f(\mathrm{OPT}) - O\Big( \tfrac{\Delta k}{\eps} \sqrt{(k + \log \tfrac{1}{\delta}) \log \tfrac{1}{\delta}} \,\log \tfrac{n}{\beta} \Big).
\]
Moreover, the algorithm makes $O(nk)$ oracle queries.
\end{theorem}
\section{Monotone Objective Functions}\label{sec:monotone}
In this section, we establish the guarantees stated in \cref{thm:two_guess_greedy_intro}. The result follows by combining the complexity, privacy, and utility guarantees from \Cref{lemma:two_guess_greedy_complexity,lemma:2guess_privacy,lemma:two_guess_greedy_utility}. All omitted proofs are provided in \Cref{appendix:monotone}.

We introduce two components. The first, formulated in \Cref{algorithm:greedy} is a DP adaptation of the density-greedy algorithm that utilizes \GenRNMT for greedy element selections. This serves as a core subroutine for the second component, \Cref{alg:two_guess_greedy}, which implements a privatized \texttt{Two-Guess-Greedy} algorithm. Specifically, \Cref{alg:two_guess_greedy} samples $p \in[0,1]$ uniformly and iterates over every subset $Y \subseteq \cN$ of size $|Y| \le 2$. For each $Y$, it performs $3/\beta$ iterations, each skipped independently with probability $p$. In non-skipped iterations where $|Y|=2$, \Cref{algorithm:greedy} is executed on the residual instance defined by the ground set $\cN \setminus Y$, remaining budget $B - c(Y)$, and marginal objective $h_D(S)=f_D(S\mid Y)$. The value of the resulting solution $S_Y \cup Y$ is then privatized using the Laplace mechanism. Finally, the candidate with the highest noisy value is returned.

\begin{algorithm}
\caption{\texttt{DP-Density-Greedy}}
\label{algorithm:greedy}
\setcounter{AlgoLine}{0}
\KwIn{
Dataset $D$;
Submodular function $h_D:2^\cN\to\R_+$;
Budget $B$;
Parameters  $\beta,\eps\in (0,1]$}

Set $S_0\gets\varnothing$ and $i\gets 0$.

\While{ $\exists u\in \cN\setminus S_i$ such that $c(S_i\cup\set{u})\le B$ \plabel{line2}}{
    Let $\cC_i$ be the set of elements satisfying the condition stated in Line \ref{line2}.
    
    Define $q_i(u;D)=\frac{h_D(u\mid S_i)}{c(u)}$ for all $u\in \cC_i$.

     Compute $u_{i+1}\gets \GenRNM{\eps}{\cC_i, q_i, \beta/k}{D}$
     
     Let $S_{i+1}\gets S_i\cup\set{u_{i+1}}$.

     Update $i\gets i+1$.
}
\Return $S_i$
\end{algorithm}
The query complexity of \Cref{alg:two_guess_greedy} is stated next.
\begin{lemma}\label[lemma]{lemma:two_guess_greedy_complexity}
    \Cref{alg:two_guess_greedy} makes $O(\beta^{-1}n^3k)$ oracle queries.
\end{lemma}

\paratitle{Privacy Analysis}
The privacy guarantee of \Cref{alg:two_guess_greedy} relies on a structural alignment with the generalized selection framework of \cite{cohen2023generalized}. This ensures a global privacy loss independent of the number of enumeration steps, complemented by a composition analysis for the individual invocations of \Cref{algorithm:greedy}.
\begin{lemma}\label[lemma]{lemma:2guess_privacy}
For any $\eps,\delta\in (0,1]$, \Cref{alg:two_guess_greedy} is $(\eps,\delta)$-DP.
\end{lemma}
\begin{proof}[Proof Sketch]
By advanced composition (\Cref{thm:composition}), each invocation of \Cref{algorithm:greedy} is $(\eps/6, \delta')$-DP, where $\delta' = O(\delta\beta/n^2)$. Consider a mechanism $\mathcal{M}_Y$ that executes \Cref{algorithm:greedy} on the residual instance defined by $Y$ and outputs $(S_Y, z_{Y,j})$ where $z_{Y,j} = f(S_Y \cup Y) + \text{Lap}(6\Delta_f/\eps_0)$. By basic composition, each $\mathcal{M}_Y$ is $(\eps/3, \delta')$-DP. Finally, \Cref{alg:two_guess_greedy} can be viewed as an instantiation of the \texttt{Selection} framework from \cite{cohen2023generalized} applied to the collection $\{\mathcal{M}_Y \mid Y \subseteq \mathcal{N}, |Y| \le 2\}$. Under our choice of parameters, the entire algorithm is $(\eps, \delta)$-DP.
\end{proof}

\begin{algorithm} 
\caption{\DPTwoGG  (\texttt{DP-Two-Guess-Greedy})}
\label{alg:two_guess_greedy}
\setcounter{AlgoLine}{0}
\SetInd{0.5em}{1em} 
\KwIn{
Dataset $D$;
Submodular function $f_D:2^\cN\to\R_+$;
Budget $B$;
Parameters  $\beta,\eps,\delta\in (0,1]$}

Sample $p \in [0,1]$ uniformly; set $\Omega \gets \varnothing$. \plabel{alg:two_guess_init}\;
 
 Let $\eps_0 \gets O(\eps / \sqrt{k \log\frac{n}{\delta\beta}})$\plabel{alg:two_guess_seteps}  \tcp*[l]{{\scriptsize Exact expression in the proof of  \Cref{lemma:2guess_privacy_formal}}}

\ForEach{$Y \subseteq \mathcal{N}$ where $|Y| \le 2, c(Y) \le B$}{
    \ForEach{$j=1,\dots,\lceil3/\beta\rceil$}{
        \WithProb{$p$}{
            \lIf{$|Y|=2$}{let $S_{Y,j}$ be the output of \Cref{algorithm:greedy} on the residual instance defined by $Y$ with parameters $\beta/3$, $\eps_0/6$. \plabel{alg2:line_residual}
                }
            \lElse{$S_{Y,j}\gets \varnothing$.}
                Let $S'_{Y,j}\gets S_{Y,j}\cup Y$
                
                Let $z_{Y,j}\gets f_D(S'_{Y,j})+ \Lap{\frac{6\Delta_f}{\eps}}$ \plabel{alg2:line_sample_laplace}

                Update  $\Omega\gets \Omega \cup \set{(Y,j)}$ 
                }
    }
}

\lIf{$\Omega \neq \varnothing$}{ let $(Y^*,j^*) \gets \argmax_{(Y,j) \in \Omega} z_{Y,j}$ and $S^* \gets S'_{Y^*,j^*}$\plabel{alg2:get_best}}
\lElse{$S^* \gets \emptyset$}

\Return $S^*$
\end{algorithm}

\paratitle{Utility Analysis}
The proof extends the analysis of \citet{feldman2023practical} to address the error introduced by our DP adaptation. Assume that $|\OPT| > 2$. The case $|\OPT| \le 2$ is addressed separately in \Cref{appendix:monotone}, as in this case the algorithm considers an optimal solution as a candidate in Line~\ref{alg2:get_best}. While density-greedy does not provide a constant-factor approximation for SMK directly, a tight $(1-1/e)$-approximation is achieved when applied to a residual instance defined by $Y^*$, a maximum-value subset of $\OPT$ of size two \cite{feldman2023practical}. Accordingly, we first analyze the execution of \Cref{algorithm:greedy} on this instance, beginning with a bound on the sensitivity of density scores.

\begin{lemma}\label[lemma]{lemma:sensitivity_bound}
For every iteration $i$ of \Cref{algorithm:greedy} and $u \in \mathcal{C}_i$, the quality score $q_i(u;D)$ has sensitivity $2\Delta_f/ c(u)$.
\end{lemma}

Let $o_m$ be an element of maximum cost in $\OPT \setminus Y^*$, which is feasible in the residual instance and is non-empty by assumption. Let $S_1,\ldots,S_\ell$ denote the sequence of partial solutions produced by \Cref{algorithm:greedy}. Define $T$ as the smallest index in $\set{1,\dots,\ell-1}$ for which the cost of $S_i$ exceeds $1 - c(Y^*) -c(o_m)$, or $T=\ell$ if no such index exists. The next lemma characterizes the progress made prior to iteration $T$ and leverages \GenRNMT to obtain an error term  that does not depend on 
the sensitivity of the density scores.

\begin{lemma}\label[lemma]{lemma:gain_progress}
Let $Y^* \subseteq \OPT$ be a maximum-value subset of size $2$. Suppose \Cref{algorithm:greedy} is instantiated with parameters $\eps$ and $\beta$ and executed on the residual instance defined by $Y^*$. Then, with probability $1-\beta$, for every $i=0,\dots, T-1$:
\begin{multline*}
(1-c(Y^*)- c(o_m))\cdot \frac{h(u_{i+1}\mid S_i)}{c(u_{i+1})} \ge \\ 
h(\OPT \setminus Y^*) - h(S_i) - h(o_m) 
- O\Big(\frac{k\Delta_f}{\eps}\log\frac{n}{\beta}\Big).
\end{multline*}
\end{lemma} 

We use \Cref{lemma:gain_progress} to show that when $S_i$ is far from optimal, each greedy step decreases the residual gap by a multiplicative factor with high probability. Unrolling this recurrence over the first $T$ iterations yields exponential decay of the residual term as a function of the accumulated cost $c(S_T)$. By the definition of $T$, the accumulated cost is sufficiently high to imply a $(1-1/e)$-approximation. 

\begin{lemma}\label[lemma]{lemma:performance_given_Y}
Suppose \Cref{algorithm:greedy} is executed on the residual instance defined by $Y^*$. Then, with probability $1-\beta$, the returned set $S_\ell$ satisfies:
$
    f(Y^* \cup S_\ell) \ge (1-\tfrac{1}{e}) \cdot f(\OPT) - O\Big(\tfrac{k \Delta_f}{\eps} \log \tfrac{n}{\beta} \Big).
$
\end{lemma}

We now derive the utility bound. Notably, the guarantee holds even when $\lvert \OPT \rvert \le 2$.
\begin{lemma}\label[lemma]{lemma:two_guess_greedy_utility}
For any $\eps, \delta>0$, $ \beta\in(0,1)$, \Cref{alg:two_guess_greedy} outputs a feasible $S^*$ such that, with probability $1-\beta$,  
\[f(S^*) \ge  (1-\tfrac{1}{e})\cdot  f(\OPT) - O\Big(\tfrac{\Delta k^{1.5}}{\eps}
\sqrt{\log\tfrac{n}{\delta\beta}}
 \log\tfrac{n}{\beta}\Big).\]
\end{lemma}
\begin{proof}[Proof sketch]
Consider an event $\cE$ where: \textit{(i)} \Cref{algorithm:greedy} is invoked for the guess $Y^*$,  \textit{(ii)} the event in \Cref{lemma:performance_given_Y} occurs, and (iii) all Laplace noises are bounded by $O(\frac{\Delta_f}{\eps_0} \log \frac{n}{\beta})$. We show $\Pr[\cE] \ge 1-\beta$. Conditioned on $\cE$, at least one candidate achieves a $(1-1/e)$-approximation up to a bounded additive error $O(\frac{k\Delta_f}{\eps_0} \log \frac{n}{\beta})$, and the final selection step (Line~\ref{alg2:get_best}) preserves this guarantee. Substituting $\eps_0$ as set in Line~\ref{alg:two_guess_seteps} yields the  bound.
\end{proof}

So far, we have presented the key ideas behind the proof of \Cref{thm:two_guess_greedy_intro}. We conclude this section with a brief overview of our approach for \Cref{thm:greedyplus_intro}, deferring the full algorithmic details and analysis to \Cref{appendix:greedyplus}. 

\paratitle{A Faster Algorithm} 
We propose \DPDGP (\Cref{algorithm:greedyplus}), a variant of \Cref{algorithm:greedy} that is executed once with the objective $f_D$, and uses \RNMT to select the best solution among: \textit{(i)} the output of \Cref{algorithm:greedy}, \textit{(ii)} all feasible single-element extensions of partial solutions observed throughout the execution of \Cref{algorithm:greedy}, and \textit{(iii)} all singletons. The utility analysis follows an argument analogous to \Cref{lemma:gain_progress}, but considers the element $o_m \in \OPT$ of maximal cost instead of $o_m\in \OPT \setminus Y^*$. This yields a progress inequality for early iterations in which the total cost does not exceed $1-c(o_m)$. Specifically, we show:
\begin{align*}
    \frac{c(u_{i+1})}{1-c(o_m)} \cdot \paren{\frac{f(\OPT)}{2} - kE} \le f(S_{i+1}) - f(S_i),
\end{align*}
where $E = O\paren{\frac{\Delta_f}{\eps_0} \log \frac{n}{\beta}}$. Summing this across all iterations leads to the bound stated in \Cref{thm:greedyplus_intro}.

\section{Non-Monotone Objective Functions}\label{sec:non-monotone}
In this section, we present our algorithm for the non-monotone setting and outline the proof for \Cref{thm:samplegreedy}. The result follows by combining the complexity, privacy, and utility guarantees established in \Cref{lemma:subsample_complexity,lemma:sample_privacy_proof_main,lemma:sample_utility_proof_main}. All omitted proofs are provided in \Cref{sec:appendix_nonmonotone}.

\Cref{algorithm:subsample_greedy} is based on the non-private \texttt{SmkRan} algorithm by \citet{han2021approximation}. \texttt{SmkRan} maintains two solution sequences, $S_j$ and $S_j^*$. In each round $j$, it selects an element $u_j^*$ with maximum marginal gain and an element $u_{j+1}$ with maximum density. $S_j^*$ is set to $S_j \cup \{u_j^*\}$ if $u_j^*$ has positive marginal gain; otherwise, $S_j^* = S_j$. Similarly, $u_{j+1}$ is added to $S_j$ with probability $1/2$ only if its marginal gain is positive. The algorithm returns the best candidate among the final $S_j$ and all encountered $S_j^*$.

Our adaptation departs from this approach in several critical aspects. First, we perform only \emph{density}-based selections using \GenRNMT. The selected element is then added to the current solution with probability $1/2$; otherwise, it is discarded. Second, instead of testing for positive marginal gains, we use \RNMT to select the best candidate among all encountered partial solutions and their feasible single-element extensions. Critically, we extend the analysis of \citet{han2021approximation} to accommodate a controlled slack for slightly negative values, and prove that elements selected in \emph{early} iterations provide sufficiently high marginal contributions with high probability. Finally, we derive a concentration bound on the total number of iterations to overcome the composition bottleneck described in \Cref{section:challenges_results_nonmonotone}.

\begin{algorithm}
\caption{ \DPSDG (\texttt{DP-Subsample-Density-Greedy})}
\label{algorithm:subsample_greedy}
\setcounter{AlgoLine}{0}

\KwIn{
Dataset $D$;
Submodular function $f_D:2^\cN\to\R_+$;
Budget $B$;
Parameters  $\beta,\eps,\delta\in (0,1]$.}

Let $\eps_0 \gets O\Big(\eps / \sqrt{(k + \log(1/\delta)) \log(1/\delta)}\Big)
$ \plabel{alg:sample_seteps} \tcp*[l]{{\scriptsize Exact expression in the proof of \Cref{lemma:sample_privacy_proof}}}

Let $S_0\gets\varnothing$ and $i\gets 0$. \plabel{alg:sample_set_init}

\While{$\exists v\in \cN\setminus \set{u_1,\dots,u_i}$ such that $c(S_i+v)\le B$ \plabel{alg:sample_line2}}{
    Let $\cC_i$ be the set of elements satisfying the condition stated in Line \ref{alg:sample_line2}.

    Define $q_i(u;D)=\frac{f_D(v\mid S_i)}{c(v)}$ for all $v\in \cC_i$.

     Compute $u_{i+1}\gets \GenRNM{\varepsilon_0}{\cC_i, q_i, \beta/(2n)}{D}$\plabel{alg:sample_GEM}

     \WithProb{$1/2$}{
        $S_{i+1}\gets S_i\cup\set{u_{i+1}}$ .\plabel{alg:sample_random_inclusion}
     }
     \lElse{
         $S_{i+1}\gets S_i$
     }
        
    Update $i\gets i+1$. \plabel{alg:sample_last_inloop}

}
 Let  $\cS \gets \set{S_{i'}\mid 0\le i'\le i } \cup \set{S_{i'}\cup\set{v} \mid 0\le i'\le i,\ v \in \cN,\ c(S_{i'}+v)\le B }$\plabel{alg:sample_defineS}
 
 Let $S^* \gets \RNM{\eps/2}{\cS, f}{D}$ \plabel{alg:sample_lastEM}
 
\Return $S^*$
\end{algorithm}

The query complexity of \Cref{algorithm:subsample_greedy} is given by the following lemma.
\begin{lemma}\label[lemma]{lemma:subsample_complexity}
    \Cref{algorithm:subsample_greedy} makes $O(nk)$ oracle queries.
\end{lemma}

\paratitle{Privacy Analysis}
The privacy guarantee relies on a concentration bound on the total number of iterations. We observe that the number of iterations between two consecutive element inclusions is stochastically dominated by a geometric random variable with parameter $1/2$. By extending this coupling to the entire execution, we establish that the total iteration count is dominated by a sum of independent geometric random variables, which concentrates strongly around $O(k)$.
\begin{lemma}\label[lemma]{lemma:runtime_bound}
    With probability at least $1-\delta$, \Cref{algorithm:subsample_greedy} executes the while loop (Lines~\ref{alg:sample_line2}--\ref{alg:sample_last_inloop}) at most $O(k+\log(1/\delta))$ times.
\end{lemma}
Intuitively, the probability that the number of iterations, and hence the number of required composition steps, exceed $O(k+\log(1/\delta))$ can be absorbed into the privacy failure probability $\delta$. Combining this with the composition theorem (\Cref{thm:composition}) yields our privacy result.
\begin{lemma}\label[lemma]{lemma:sample_privacy_proof_main}
   For any $\eps,\delta\in (0,1]$, \Cref{algorithm:subsample_greedy} is $(\eps,\delta)$-DP.
\end{lemma}

\paratitle{Utility Analysis}
Our approach extends the analysis of \citet{han2021approximation}. We first define a high-probability event $\cE$ under which the DP-induced error remains bounded. The formal definition of $\cE$ is deferred to the appendix, where we prove that $\Pr[\cE] \ge 1 - \beta$. In particular, the occurrence of $\cE$ guarantees that:
\begin{itemize}
\item[\textit{(i)}] 
For every $v \in \cC_i$, the element $u_{i+1}$  obtained in Line~\ref{alg:sample_GEM} satisfies $\frac{f_D(u_{i+1}|S_i)}{c(u_{i+1})} \ge \frac{f_D(v|S_i)}{c(v)} - O\left(\frac{\Delta_f}{c(v)\eps_0}\log\frac{n}{\beta}\right),$
\item[\textit{(ii)}] the set $S^*$ selected in Line~\ref{alg:sample_lastEM} satisfies $f(S^*) \ge \max_{S\in \cS} f(S) - O(\frac{\Delta_f}{\eps}\log\frac{n}{\beta})$.
\end{itemize}
We analyze the expected utility conditioned on the occurrence of $\cE$, starting by introducing notation. Let $o_m$ be an element of $\OPT$ with maximal cost. Let $S_0, \dots, S_r$ be the partial solutions produced by \Cref{algorithm:subsample_greedy}, where $r$ is the earliest iteration such that either the algorithm terminates, or there is no element $u \in \OPT \setminus S_r$ such that $S_r \cup \{u\}$ is feasible and $f(u \mid S_r) = \Omega(\frac{\Delta_f}{\eps} \log \frac{n}{\beta})$. Additionally, let $T$ be the smallest index $i < r$ such that $c(S_i \cup \{u_{i+1}\}) > c(\OPT \setminus \{o_m\})$, and set $T = r$ if no such index exists. We consider two disjoint subsets of $\OPT$: $O_{\le T}$ contains elements selected as $u_i$ for $i\le T$, and $O_{>T}$ contains elements not in $O_{\le T}$ with marginal contribution to $S_T$ at least \(\Omega(\frac{\Delta_f}{\eps_0} \log\frac{n}{\beta})\).

The following lemma bounds the total marginal contribution from elements in $\OPT \setminus \{o_m\}$ in terms of marginal gains from selected elements, discarded $\OPT$ elements, and a DP-induced error term. Crucially, it leverages \Cref{thm:genRNM_guarantee} to ensure this error remains independent of the sensitivities of the density scores. The analysis relies on a delicate fractional mapping from \citet{han2021approximation}, which allocates the cost of $O_{>T} \setminus \{o_m\}$ elements onto elements of the algorithm’s partial solution. This mapping allows the proof to upper-bound the total marginal gain of those optimal elements in terms of selected elements.
\begin{lemma}\label[lemma]{lemma:gain_bound}
    For each $u\in V$, let $X_u$ be an indicator for the event that $u\in O_{\le T}\setminus (S_T\cup\set{o_m})$ or $u\in S_T \setminus (\OPT \setminus\set{o_m})$. Then,
    \begin{align*}
f(S_T \cup \OPT) &\le f(S_T \cup \{o_m\}) + \sum_{u\in\cN} X_u f(u\mid S_u) \\
&\quad + f(S^*)-f(S_T) + O\Big(\tfrac{k\Delta_f}{\eps_0}\log \tfrac{n}{\beta}\Big).
\end{align*}
\end{lemma}
The utility guarantee of \Cref{algorithm:subsample_greedy} is stated next.
\begin{lemma}\label[lemma]{lemma:sample_utility_proof_main}
For any $\eps, \delta>0$, $ \beta\in(0,1)$, \Cref{algorithm:subsample_greedy} outputs a feasible set $S^*$ such that, with probability $1-\beta$,
\[
    \Exp[f(S^*)] \ge \tfrac{1}{4}f(\OPT) - O\Big(\tfrac{\Delta k}{\varepsilon} \sqrt{(k + \log \tfrac{1}{\delta}) \log \tfrac{1}{\delta}} \,\log \tfrac{n}{\beta}\Big).
\]
\end{lemma}
\begin{proof}[Proof Sketch]
We leverage a lemma from \cite{han2021approximation} to show that 
$
\Exp[f(S_T)] = \Exp[\sum_{u\in \cN} X_u \cdot f(u\mid S_u)].
$
Combining this equality with \Cref{lemma:gain_bound} yields the upper-bound
\begin{equation*}
    \Exp[f(S_T \cup \OPT)] \le 2\Exp[f(S^*)] + O\!\left(\tfrac{k\Delta_f}{\eps_0}\log\tfrac{n}{\beta}\right).
\end{equation*}
   On the other hand, a key lemma of \citet{buchbinder2014submodular}  for non-monotone submodular functions implies the lower-bound $\Exp[f(S_T \cup \OPT)] \ge \tfrac{1}{2} f(\OPT)$, as each element is included in $S_T$ with probability at most $1/2$. Rearranging these inequalities and substituting $\eps_0$ as set in Line~\ref{alg:sample_seteps} completes the proof.
\end{proof}

\begin{figure}
\centering
\includegraphics[width=0.48\linewidth]{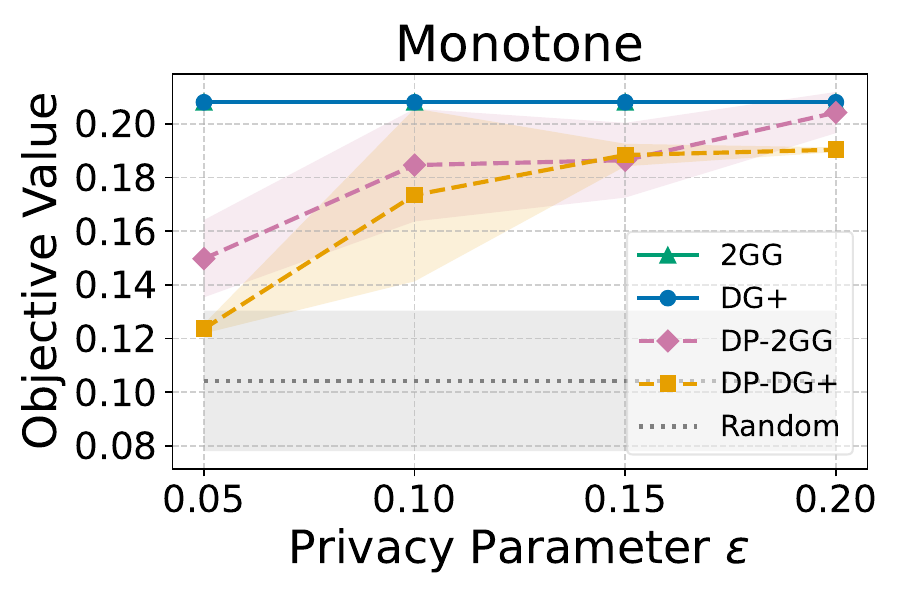}
\includegraphics[width=0.48\linewidth]{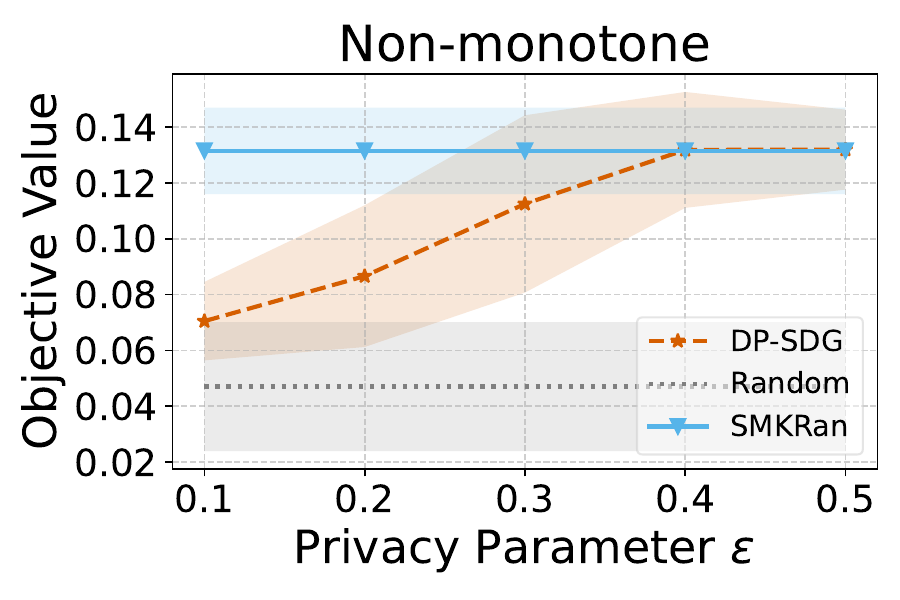}
\caption{Monotone (left) and non-monotone (right) objective value for varying $\eps$.}
\label{fig:eps_obj}
\end{figure}

\begin{figure}
\centering
\includegraphics[width=0.48\linewidth]{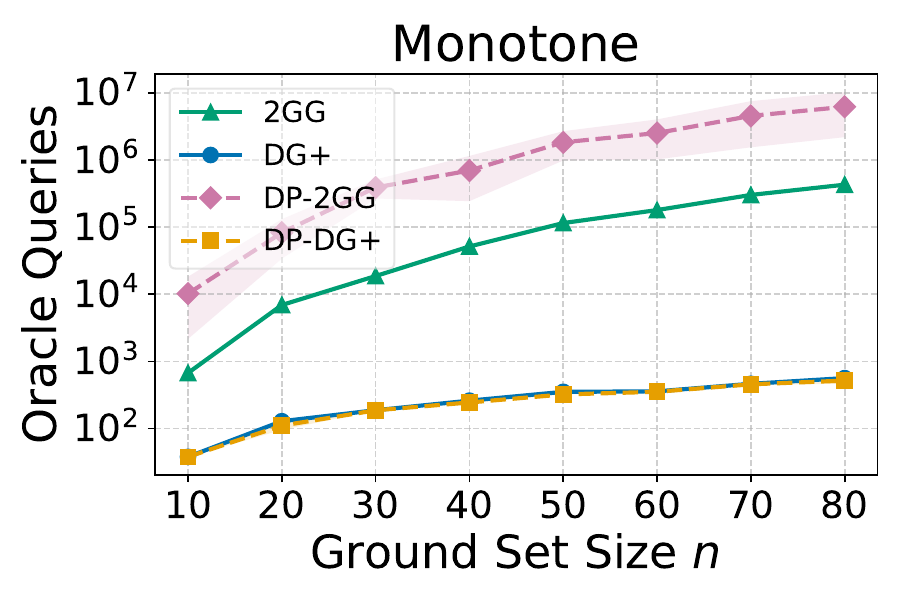}
\includegraphics[width=0.48\linewidth]{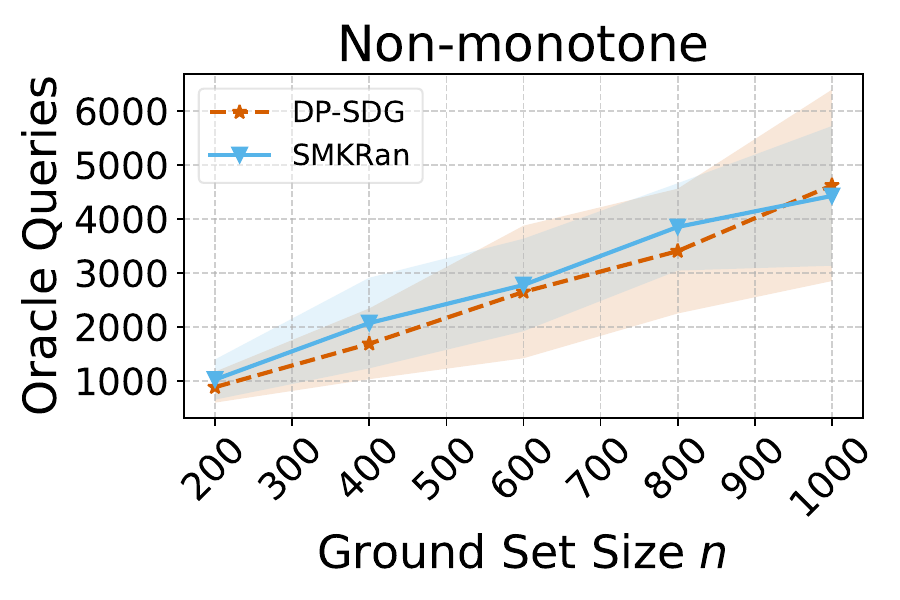}
\caption{Monotone (left, log-scale) and non-monotone (right) number of queries for varying $n$.}
\label{fig:n_queries}
\end{figure}

\section{Empirical Demonstration}
\label{sec:experiments}
We demonstrate the effectiveness of our proposed algorithms on a ride-sharing optimization task, a standard benchmark from prior work~\cite{mitrovic2017differentially,chaturvedi2021differentially}.
 We examine the effect of $\eps$ on utility (\Cref{fig:eps_obj}) and the effect of $n$ on the number of oracle queries (\Cref{fig:n_queries}). \Cref{appendix:experiments} includes results on the impact of the budget $B$ and higher-scale results for \DPDGP. Our implementation is available online.\footnote{\url{https://github.com/ronzadi/dp-smk}.}

\paratitle{Dataset}
Following \citet{mitrovic2017differentially, mitrovic2018data}, we use a dataset $D$ of $m=100{,}000$ Uber pickup locations in Manhattan~\cite{uberDataset}, where each record is a $(\text{latitude}, \text{longitude})$ coordinate
pair. The goal is to select a subset of waiting locations for idle drivers from a grid $\cN$ of $n$ candidate points while ensuring differential privacy for individuals in the dataset.\footnote{Assuming each pickup corresponds to a distinct individual.} We further assign each candidate $b\in \cN$  a cost
$
c(b) = 1 + \frac{19}{\max(d_{\mathrm{km}}(b,\, p_{\text{PoI}}), 1)} \in [1,20],
$
where $d_{\mathrm{km}}$ denotes distance in kilometers and $p_{\text{PoI}}$ is a fixed point in Midtown Manhattan.\footnote{Specifically, we use Penn Station as the reference location.} This induces a location premium that decays with distance from the city center.

\paratitle{Objective Functions}
For a pickup point $a=(x_a,y_a)$ and a grid point $b=(x_b,y_b)$, we adopt the convenience score $s(a,b) = 2 - \frac{2}{1 + e^{-200 \norm{a-b}_1}}$ from \citet{mitrovic2018data}. In the monotone setting, we maximize $f_D(S) = \frac{1}{m} \sum_{a \in D} \max_{b \in S} s(a,b)$, a monotone submodular function with sensitivity $1/m$. While $f_D$ measures coverage, it does not reward spatial diversity and may yield solutions concentrated in a small area. Thus, in the non-monotone setting, we maximize $g_D(S) =f_D(S) + \frac{\lambda}{nk} \sum_{b' \in \mathcal{N} \setminus S} \sum_{b \in S} s(b', b)$, fixing $\lambda = 0.1$. Since the added diversity term is a graph-cut function independent of $D$, $g_D$ is a non-monotone submodular function that retains sensitivity $1/m$.

\paratitle{Algorithms}
In the monotone setting, we evaluate \DPTwoGG (\Cref{alg:two_guess_greedy}) and \DPDGP (\Cref{algorithm:greedyplus}) against their non-private counterparts \texttt{Two-Guess-Greedy} (\TwoGG)~\citep{feldman2023practical} and \texttt{Density-Greedy$^+$} (\DGP)~\citep{yaroslavtsev2020bring}, while in the non-monotone setting we compare \DPSDG (\Cref{algorithm:subsample_greedy}) with the non-private \SmkRan~\cite{han2021approximation}. In both settings, we include \Random, a trivial DP baseline that processes elements in a random order, adding each to the solution if it maintains feasibility.

\paratitle{Experimental Settings}\label{sec:settings}
We use by default $\eps=1$ and $\delta = m^{-1.5}$. \DPSDG and \DPDGP are evaluated with defaults $n=100, B=20$ ($k=8$). To accommodate the more query-intensive \DPTwoGG, experiments including it use a default of $n=30, B=40$ ($k=12$). While we have used advanced composition in our theoretical analyses to yield asymptotically lower noise, constant factors may favor basic composition in some regimes. We set the per-iteration privacy parameter $\eps_0$ according to the composition rule (basic or advanced) that achieves $(\eps,\delta)$-DP with the smallest noise scale. All results are averaged over 10 runs and shaded regions indicate one standard deviation.

\paratitle{Results}
\Cref{fig:eps_obj} shows that \DPTwoGG and \DPDGP achieve utility comparable to their non-private counterparts, substantially outperforming \Random for $\eps \ge 0.1$. At $\eps=0.2$, \DPTwoGG is within 2\% of \TwoGG and \DPDGP within 8\%, whereas \Random is 49\% lower. Moreover, \DPSDG attains utility close to that of the non-private \SmkRan; at $\eps=0.4$, it is within 0.1\%, while \Random is 61\% lower.
\Cref{fig:n_queries} shows that \DPDGP is significantly more scalable than \DPTwoGG, which requires $6000\times$ more queries at $n=60$. Furthermore, \DPSDG exhibits linear growth in the number of queries with~$n$, underscoring its practical efficiency. Similar trends for \DPDGP are reported in~\Cref{appendix:experiments}.

\section{Conclusion}
We presented differentially private algorithms for maximizing monotone and non-monotone submodular functions subject to a knapsack constraint. Our results improve both utility and query complexity in the monotone setting and establish the first provable guarantees for the non-monotone setting.
There are several interesting directions for future research. A primary question is whether an error term with better dependence on $k$ is achievable with polynomial query complexity under the bounded sensitivity assumption, a problem that remains unresolved even for cardinality constraints. Furthermore, while a $(1/2-\eta)$-approximation can be achieved using $O_{\eta}(n)$ queries in the non-private setting via decreasing thresholds~\cite{li2022submodular}, adapting such techniques to the DP setting without worsening the dependence on $k$ appears non-trivial. Finally, extending our techniques to settings combining knapsack and matroid constraints~\cite{badanidiyuru2014fast,mirzasoleiman2016fast} is a natural direction for future work.
\FloatBarrier
\bibliography{src/bibl}

@article{feldman2023practical,
  title={Practical budgeted submodular maximization},
  author={Feldman, Moran and Nutov, Zeev and Shoham, Elad},
  journal={Algorithmica},
  volume={85},
  number={5},
  pages={1332--1371},
  year={2023},
  publisher={Springer}
}

@inproceedings{mcsherry2007mechanism,
  title={Mechanism design via differential privacy},
  author={McSherry, Frank and Talwar, Kunal},
  booktitle={48th Annual IEEE Symposium on Foundations of Computer Science (FOCS'07)},
  pages={94--103},
  year={2007},
  organization={IEEE}
}

@article{dwork2014algorithmic,
  title={The algorithmic foundations of differential privacy},
  author={Dwork, Cynthia and Roth, Aaron and others},
  journal={Foundations and Trends{\textregistered} in Theoretical Computer Science},
  volume={9},
  number={3--4},
  pages={211--407},
  year={2014},
  publisher={Now Publishers, Inc.}
}

@article{han2021approximation,
  title={Approximation algorithms for submodular data summarization with a knapsack constraint},
  author={Han, Kai and Cui, Shuang and Zhu, Tianshuai and Zhang, Enpei and Wu, Benwei and Yin, Zhizhuo and Xu, Tong and Tang, Shaojie and Huang, He},
  journal={Proceedings of the ACM on Measurement and Analysis of Computing Systems},
  volume={5},
  number={1},
  pages={1--31},
  year={2021},
  publisher={ACM New York, NY, USA}
}

@inproceedings{buchbinder2014submodular,
  title={Submodular maximization with cardinality constraints},
  author={Buchbinder, Niv and Feldman, Moran and Naor, Joseph and Schwartz, Roy},
  booktitle={Proceedings of the twenty-fifth annual ACM-SIAM symposium on Discrete algorithms},
  pages={1433--1452},
  year={2014},
  organization={SIAM}
}

@article{mckenna2020permute,
  title={Permute-and-flip: A new mechanism for differentially private selection},
  author={McKenna, Ryan and Sheldon, Daniel R},
  journal={Advances in Neural Information Processing Systems},
  volume={33},
  pages={193--203},
  year={2020}
}

@article{ding2021permute,
  title={The permute-and-flip mechanism is identical to report-noisy-max with exponential noise},
  author={Ding, Zeyu and Kifer, Daniel and Steinke, Thomas and Wang, Yuxin and Xiao, Yingtai and Zhang, Danfeng and others},
  journal={arXiv preprint arXiv:2105.07260},
  year={2021}
}

@article{janson2018tail,
  title={Tail bounds for sums of geometric and exponential variables},
  author={Janson, Svante},
  journal={Statistics \& Probability Letters},
  volume={135},
  pages={1--6},
  year={2018},
  publisher={Elsevier}
}

@inproceedings{krause2005near,
  title={Near-optimal nonmyopic value of information in graphical models},
  author={Krause, Andreas and Guestrin, Carlos},
  booktitle={Proceedings of the Twenty-First Conference on Uncertainty in Artificial Intelligence},
  pages={324--331},
  year={2005}
}

@inproceedings{wei2015submodularity,
  title={Submodularity in data subset selection and active learning},
  author={Wei, Kai and Iyer, Rishabh and Bilmes, Jeff},
  booktitle={International conference on machine learning},
  pages={1954--1963},
  year={2015},
  organization={PMLR}
}

@inproceedings{dwork2006calibrating,
  title={Calibrating noise to sensitivity in private data analysis},
  author={Dwork, Cynthia and McSherry, Frank and Nissim, Kobbi and Smith, Adam},
  booktitle={Theory of Cryptography Conference},
  pages={265--284},
  year={2006},
  organization={Springer}
}

@inproceedings{gupta2010differentially,
  title={Differentially private combinatorial optimization},
  author={Gupta, Anupam and Ligett, Katrina and McSherry, Frank and Roth, Aaron and Talwar, Kunal},
  booktitle={Proceedings of the twenty-first annual ACM-SIAM symposium on Discrete Algorithms},
  pages={1106--1125},
  year={2010},
  organization={SIAM}
}

@inproceedings{mitrovic2017differentially,
  title={Differentially private submodular maximization: Data summarization in disguise},
  author={Mitrovic, Marko and Bun, Mark and Krause, Andreas and Karbasi, Amin},
  booktitle={International Conference on Machine Learning},
  pages={2478--2487},
  year={2017},
  organization={PMLR}
}

@inproceedings{lakkaraju2016interpretable,
  title={Interpretable decision sets: A joint framework for description and prediction},
  author={Lakkaraju, Himabindu and Bach, Stephen H and Leskovec, Jure},
  booktitle={Proceedings of the 22nd ACM SIGKDD international conference on knowledge discovery and data mining},
  pages={1675--1684},
  year={2016}
}

@inproceedings{feldman2009private,
  title={Private coresets},
  author={Feldman, Dan and Fiat, Amos and Kaplan, Haim and Nissim, Kobbi},
  booktitle={Proceedings of the forty-first annual ACM symposium on Theory of computing},
  pages={361--370},
  year={2009}
}

@inproceedings{dwork2006differential,
  title={Differential privacy},
  author={Dwork, Cynthia},
  booktitle={International colloquium on automata, languages, and programming},
  pages={1--12},
  year={2006},
  organization={Springer}
}

@inproceedings{chaturvedi2021differentially,
  title={Differentially private decomposable submodular maximization},
  author={Chaturvedi, Anamay and L{\^e} Nguy{\~e}n, Huy and Zakynthinou, Lydia},
  booktitle={Proceedings of the AAAI Conference on Artificial Intelligence},
  volume={35},
  number={8},
  pages={6984--6992},
  year={2021}
}

@article{nemhauser1978analysis,
  title={An analysis of approximations for maximizing submodular set functions—I},
  author={Nemhauser, George L and Wolsey, Laurence A and Fisher, Marshall L},
  journal={Mathematical programming},
  volume={14},
  pages={265--294},
  year={1978},
  publisher={Springer}
}

@inproceedings{chaturvedi2023streaming,
  title={Streaming submodular maximization with differential privacy},
  author={Chaturvedi, Anamay and Nguyen, Huy and Nguyen, Thy Dinh},
  booktitle={International Conference on Machine Learning},
  pages={4116--4143},
  year={2023},
  organization={PMLR}
}

@inproceedings{dwork2010boosting,
  title={Boosting and differential privacy},
  author={Dwork, Cynthia and Rothblum, Guy N and Vadhan, Salil},
  booktitle={2010 IEEE 51st annual symposium on foundations of computer science},
  pages={51--60},
  year={2010},
  organization={IEEE}
}

@inproceedings{rafiey2020fast,
  title={Fast and private submodular and $ k $-submodular functions maximization with matroid constraints},
  author={Rafiey, Akbar and Yoshida, Yuichi},
  booktitle={International conference on machine learning},
  pages={7887--7897},
  year={2020},
  organization={PMLR}
}

@inproceedings{mirzasoleiman2016fast,
  title={Fast constrained submodular maximization: Personalized data summarization},
  author={Mirzasoleiman, Baharan and Badanidiyuru, Ashwinkumar and Karbasi, Amin},
  booktitle={International Conference on Machine Learning},
  pages={1358--1367},
  year={2016},
  organization={PMLR}
}

@inproceedings{sadeghi2021differentially,
  title={Differentially private monotone submodular maximization under matroid and knapsack constraints},
  author={Sadeghi, Omid and Fazel, Maryam},
  booktitle={International Conference on Artificial Intelligence and Statistics},
  pages={2908--2916},
  year={2021},
  organization={PMLR}
}

@inproceedings{cohen2023generalized,
  title={Generalized Private Selection and Testing with High Confidence},
  author={Cohen, Edith and Lyu, Xin and Nelson, Jelani and Sarl{\'o}s, Tam{\'a}s and Stemmer, Uri},
  booktitle={14th Innovations in Theoretical Computer Science Conference (ITCS 2023)},
  pages={39--1},
  year={2023},
  organization={Schloss Dagstuhl--Leibniz-Zentrum f{\"u}r Informatik}
}

@inproceedings{yaroslavtsev2020bring,
  title={“bring your own greedy”+ max: near-optimal 1/2-approximations for submodular knapsack},
  author={Yaroslavtsev, Grigory and Zhou, Samson and Avdiukhin, Dmitrii},
  booktitle={International Conference on Artificial Intelligence and Statistics},
  pages={3263--3274},
  year={2020},
  organization={PMLR}
}

@article{amanatidis2020fast,
  title={Fast adaptive non-monotone submodular maximization subject to a knapsack constraint},
  author={Amanatidis, Georgios and Fusco, Federico and Lazos, Philip and Leonardi, Stefano and Reiffenh{\"a}user, Rebecca},
  journal={Advances in neural information processing systems},
  volume={33},
  pages={16903--16915},
  year={2020}
}

@article{sviridenko2004note,
  title={A note on maximizing a submodular set function subject to a knapsack constraint},
  author={Sviridenko, Maxim},
  journal={Operations Research Letters},
  volume={32},
  number={1},
  pages={41--43},
  year={2004},
  publisher={Elsevier}
}

@article{khuller1999budgeted,
  title={The budgeted maximum coverage problem},
  author={Khuller, Samir and Moss, Anna and Naor, Joseph Seffi},
  journal={Information processing letters},
  volume={70},
  number={1},
  pages={39--45},
  year={1999},
  publisher={Elsevier}
}

@article{wolsey1982maximising,
  title={Maximising real-valued submodular functions: Primal and dual heuristics for location problems},
  author={Wolsey, Laurence A},
  journal={Mathematics of Operations Research},
  volume={7},
  number={3},
  pages={410--425},
  year={1982},
  publisher={INFORMS}
}

@article{buchbinder2019constrained,
  title={Constrained submodular maximization via a nonsymmetric technique},
  author={Buchbinder, Niv and Feldman, Moran},
  journal={Mathematics of Operations Research},
  volume={44},
  number={3},
  pages={988--1005},
  year={2019},
  publisher={INFORMS}
}

@article{kulik2013approximations,
  title={Approximations for monotone and nonmonotone submodular maximization with knapsack constraints},
  author={Kulik, Ariel and Shachnai, Hadas and Tamir, Tami},
  journal={Mathematics of Operations Research},
  volume={38},
  number={4},
  pages={729--739},
  year={2013},
  publisher={INFORMS}
}

@inproceedings{ghazi2024individualized,
  title={Individualized privacy accounting via subsampling with applications in combinatorial optimization},
  author={Ghazi, Badih and Kamath, Pritish and Kumar, Ravi and Manurangsi, Pasin and Sealfon, Adam},
  booktitle={Proceedings of the 41st International Conference on Machine Learning},
  pages={15491--15511},
  year={2024}
}

@inproceedings{kempe2003maximizing,
  title={Maximizing the spread of influence through a social network},
  author={Kempe, David and Kleinberg, Jon and Tardos, {\'E}va},
  booktitle={Proceedings of the ninth ACM SIGKDD international conference on Knowledge discovery and data mining},
  pages={137--146},
  year={2003}
}

@inproceedings{gomes2010budgeted,
  title={Budgeted Nonparametric Learning from Data Streams.},
  author={Gomes, Ryan and Krause, Andreas},
  booktitle={ICML},
  volume={1},
  pages={3},
  year={2010}
}

@article{krause2008near,
  title={Near-optimal sensor placements in Gaussian processes: Theory, efficient algorithms and empirical studies.},
  author={Krause, Andreas and Singh, Ajit and Guestrin, Carlos},
  journal={Journal of Machine Learning Research},
  volume={9},
  number={2},
  year={2008}
}

@inproceedings{esfandiari2021adaptivity,
  title={Adaptivity in adaptive submodularity},
  author={Esfandiari, Hossein and Karbasi, Amin and Mirrokni, Vahab},
  booktitle={Conference on Learning Theory},
  pages={1823--1846},
  year={2021},
  organization={PMLR}
}

@article{li2022submodular,
  title={Submodular maximization in clean linear time},
  author={Li, Wenxin and Feldman, Moran and Kazemi, Ehsan and Karbasi, Amin},
  journal={Advances in neural information processing systems},
  volume={35},
  pages={17473--17487},
  year={2022}
}

@inproceedings{raskhodnikova2015efficient,
  title={Lipschitz extensions for node-private graph statistics and the generalized exponential mechanism},
  author={Raskhodnikova, Sofya and Smith, Adam},
  booktitle={2016 IEEE 57th Annual Symposium on Foundations of Computer Science (FOCS)},
  pages={495--504},
  year={2016},
  organization={IEEE}
}

@inproceedings{liu2019private,
  title={Private selection from private candidates},
  author={Liu, Jingcheng and Talwar, Kunal},
  booktitle={Proceedings of the 51st Annual ACM SIGACT Symposium on Theory of Computing},
  pages={298--309},
  year={2019}
}

@inproceedings{
papernot2022hyperparameter,
title={Hyperparameter Tuning with Renyi Differential Privacy},
author={Nicolas Papernot and Thomas Steinke},
booktitle={International Conference on Learning Representations},
year={2022}
}

@inproceedings{mitrovic2018data,
  title={Data summarization at scale: A two-stage submodular approach},
  author={Mitrovic, Marko and Kazemi, Ehsan and Zadimoghaddam, Morteza and Karbasi, Amin},
  booktitle={International Conference on Machine Learning},
  pages={3596--3605},
  year={2018},
  organization={PMLR}
}

@misc{uberDataset,
  author       = {{FiveThirtyEight}},
  title        = {Uber Pickups in New York City},
  year         = {2014},
  howpublished = {\url{https://www.kaggle.com/fivethirtyeight/uber-pickups-in-new-york-city}},
}

@inproceedings{badanidiyuru2014fast,
  title={Fast algorithms for maximizing submodular functions},
  author={Badanidiyuru, Ashwinkumar and Vondr{\'a}k, Jan},
  booktitle={Proceedings of the twenty-fifth annual ACM-SIAM symposium on Discrete algorithms},
  pages={1497--1514},
  year={2014},
  organization={SIAM}
}
\bibliographystyle{icml2026}

\newpage
\appendix
\onecolumn
\appendix
\crefalias{section}{appendix}
\section{Generalized Private Selection}\label{appendix:private_selection}
In this section, we provide the technical background for the generalized DP selection mechanisms utilized throughout this work. We first address the problem of private selection under varying sensitivities in \Cref{sec:appendix_RNM}. Subsequently, in \Cref{sec:private_candidates}, we discuss the selection framework used to select the best candidate from a collection of DP releases.

\subsection{Report Noisy Max for Scores with Varying Sensitivity}
\label{sec:appendix_RNM}
Consider a finite set $\mathcal{C}$, where each element $u \in \mathcal{C}$ is associated with a quality score $q(u;D) : \mathcal{X}^m \to \mathbb{R}$. A fundamental task in differential privacy is the private selection of an element $u \in \mathcal{C}$ that approximately maximizes the quality score $q(u;D)$ given a sensitive dataset $D \in \mathcal{X}^m$. One of the most widely used and computationally efficient mechanisms for this task is the \emph{Report Noisy Max} mechanism with exponential noise~\cite{dwork2014algorithmic}, denoted $\RNM{\eps}{\cC,q}{D}$. The mechanism perturbs each score with independent exponential noise and returns the candidate with the highest perturbed value. 

\begin{remark} Report Noisy Max with exponential noise has been shown to be equivalent to the \texttt{Permute-and-Flip} mechanism~\cite{ding2021permute}, which dominates the Exponential Mechanism in terms of utility~\cite{mckenna2020permute}. While we use exponential noise for this reason, our analysis is not strictly dependent on this choice; other distributions, such as Laplace noise~\cite{dwork2014algorithmic}, would yield similar results. \end{remark}

\begin{algorithm}[H]
\caption{$\RNM{\varepsilon}{\cC, q}{D}$ \cite{dwork2014algorithmic}}
\setcounter{AlgoLine}{0}
\KwIn{Dataset $D\in \cX^m$; candidate set $\cC$; for each $u\in \cC$ a score function $q(u;D):\cX^m\to\R$ with sensitivity $\Delta_u$; privacy parameter $\eps>0$.} 

$\Delta \gets \max_{u\in \cC} \Delta_u$\;

\For{$u\in \cC$}{
    Sample $Z_u \gets \Expo{\frac{\eps}{2\Delta}}$ independently\;
}
\Return $\arg\max_{u\in\cC} \{q(u;D)+Z_u\}$.
\end{algorithm}

\begin{theorem}[Restatement of \Cref{thm:RNM_main}, \citealp{dwork2014algorithmic}]\label{thm:RNM_formal}
For any $\eps>0$, $\RNM{\eps}{\cC,q}{D}$ satisfies $\eps$-differential privacy. Moreover, for any $\beta\in(0,1)$, with probability at least $1-\beta$, the output $u^*$ satisfies
\[
q(u^*;D)\ge \max_{u\in\cC} q(u;D) - \frac{2\Delta}{\eps}\log\frac{|\cC|}{\beta}.
\]
\end{theorem}

As the bound shows, the additive error of standard \RNMT scales with the maximum sensitivity across all candidates. In our setting, scores arise from density-based greedy selections, and this bound introduces a suboptimal dependence on element costs, as discussed in \Cref{section:challenges_results}. This motivates the use of a more refined private selection mechanism that adapts to individual sensitivities.

We employ the generalized selection mechanism of~\cite{raskhodnikova2015efficient}, which we denote by $\GenRNM{\varepsilon}{\cC,q,\beta}{D}$. The mechanism takes an additional failure probability parameter $\beta$ and internally invokes \RNMT on carefully constructed auxiliary scores. These scores are designed to reduce sensitivity while appropriately penalizing candidates with large individual sensitivities. 
Next, we provide the pseudocode and review the utility proof. 

\begin{algorithm}[H]
\caption{$\GenRNM{\varepsilon}{\cC, q, \beta}{D}$ \cite{raskhodnikova2015efficient}}
\LinesNumbered
\setcounter{AlgoLine}{0}
\KwIn{Dataset $D\in \cX^m$; candidate set $\cC$; for each $u\in \cC$ a score function $q(u;D):\cX^m\to\R$ with sensitivity $\Delta_u$; parameters $\eps>0$ and $\beta\in(0,1)$.} 

Set $t \gets \frac{2}{\eps}\log\frac{|\cC|}{\beta}$ 

\For{$u\in \cC$}{
    Define 
    $
    q'(u;D) \gets \min_{v\in\cC}
    \frac{(q(u;D)-t\Delta_u)-(q(v;D)-t\Delta_v)}{\Delta_u+\Delta_v}
    $\;
    
    Sample $Z_u \gets \Expo{\frac{\eps}{2}}$ independently\;
}

\Return $\arg\max_{u\in\cC} \{q'(u;D)+Z_u\}$.
\end{algorithm}

\begin{theorem}[Formal version of \Cref{thm:genRNM_guarantee}, \citealp{raskhodnikova2015efficient}]\label{thm:GRNM_formal}
Suppose that for each $u\in\cC$, the score function $q(u;D):\cX^m\to\R$ has sensitivity $\Delta_u$ (with respect to its dataset argument). For any $\eps>0$ and $\beta\in(0,1)$, the mechanism $\GenRNM{\varepsilon}{\cC,q,\beta}{D}$ is $\eps$-DP. Furthermore, with probability at least $1-\beta$, the output $u^*$ satisfies
\[
q(u^*;D)\ge \max_{u\in\cC}
\left\{
q(u;D) - \frac{4\Delta_u}{\eps}\log\frac{|\cC|}{\beta}
\right\}.
\]
\end{theorem}
\begin{proof}
The privacy guarantee follows directly from the guarantee of \RNMT (\Cref{thm:RNM_formal}), since for each $u \in \cC$, the score $q'(u;D)$ has sensitivity at most~$1$ (with respect to its dataset argument).

By a union bound and the tail of the exponential distribution, with probability at least $1-\beta$ all noise variables are bounded by
$t = \tfrac{2}{\eps}\log\tfrac{|\cC|}{\beta}$.
Condition on this event.
Let $u^*$ denote the selected element. Then for any $v \in \cC$,
\[
q'(u^*;D) \ge q'(v;D) - t,
\]
since otherwise $q'(v;D) + Z_v > q'(u^*;D) + Z_{u^*}$, contradicting the choice of $u^*$.

Let $v'\in \cC$ be the element maximizing $q(v';D) - t\Delta_{v'}$. For all $u \in \cC$, we have
$q(v';D) - t\Delta_{v'} \ge q(u;D) - t\Delta_u$, which implies
\[
q'(v';D)
= \min_{u \in \cC}
\frac{(q(v';D) - t\Delta_{v'}) - (q(u;D) - t\Delta_u)}{\Delta_{v'} + \Delta_u}
\ge 0.
\]
Combining the inequalities yields $q'(u^*;D) \ge -t$.

Finally, for any $u \in \cC$, the definition of $q'$ gives
\[
q'(u^*;D)
\le
\frac{(q(u^*;D) - t\Delta_{u^*}) - (q(u;D) - t\Delta_u)}{\Delta_{u^*} + \Delta_u}.
\]
Together with $q'(u^*;D) \ge -t$, this implies
\[
q(u^*;D) \ge q(u;D) - 2t\Delta_u,
\]
completing the proof.
\end{proof}

\begin{remark}\label[remark]{remark;GRNM}
As the previous proof shows, instead of setting $t = \tfrac{2}{\eps}\log\tfrac{|\cC|}{\beta}$, one may use any upper bound $M \ge |\cC|$, in which case, with probability at least $1-\beta$, the output $u^*$ satisfies
\[
q(u^*;D)\ge \max_{u\in\cC}
\left\{
q(u;D) - \frac{4\Delta_u}{\eps}\log\frac{M}{\beta}
\right\}.
\]
In \Cref{sec:appendix_nonmonotone}, we take
$t = \tfrac{2}{\eps}\log\tfrac{n}{\beta}$,
where $n$ denotes the size of the ground set, as this choice is used in our utility analysis in that section. Specifically, it allows us to sample all DP noises upfront, and focus our analysis on a randomized execution of    \Cref{algorithm:subsample_greedy} with fixed realization for the DP mechanisms. Note that all candidate sets $\cC_i$ arising in \Cref{algorithm:subsample_greedy} satisfy $|\cC_i|\le n$. 
\end{remark}

\subsection{Generalized Private Selection}
\label{sec:private_candidates}

A more general setting than the one considered in \Cref{sec:appendix_RNM} arises when we have a collection of $k$ mechanisms $\{\cM_i\}_{i=1}^k$, each producing outputs in an ordered domain (e.g., scored solutions), and the goal is to select an approximately best solution produced by any $\cM_i$, while incurring a privacy cost close to that of executing a single $\cM_i$ and releasing its output. This problem is commonly referred to as \emph{generalized private selection} \cite{liu2019private,papernot2022hyperparameter,cohen2023generalized}, and it has found various applications in settings such as parameter tuning in machine learning. In this work, we adopt the framework of \cite{cohen2023generalized} due to its simplicity and high utility when combined with our analysis.

The algorithm $\Selection_{\tau,\gamma}$ is parameterized by $\gamma \in (0,\infty)$ and $\tau \in \mathbb{N}$. Given $k$ private mechanisms $\{\cM_i\}_{i=1}^k$ producing scored outputs, it first samples $
p \in [0,1]
$
where $\Pr[p \le x] = x^\gamma$ for all $x \in [0,1]$. Then, for each $i \in [k]$, it runs $\cM_i$ a random number of times drawn as $\mathrm{Bin}(\tau, p)$. Finally, it outputs the highest scored candidate observed across all invocations. The pseudocode is as follows:

\begin{algorithm}[H]
\caption{$\Selection_{\tau,\gamma}$ \cite{cohen2023generalized}}\label{alg:selection}
\LinesNumbered  
\setcounter{AlgoLine}{0}
\KwIn{Dataset $D$; Private mechanisms $\set{\cM_i}_{i=1}^k$} 
 Sample $p\in[0,1]$ where $\Pr[p\le x^\gamma]=x^\gamma$ for all $x\in[0,1]$\;

 $S \gets \varnothing$\;

\For{$i = 1,\dots,k$}{
    \For{$j = 1,\dots,\tau$}{
        \WithProb{$p$}{
            $s_{i,j} \gets \cM_i(D)$\;
            
            $S \gets S \cup \{s_{i,j}\}$\;
        }
    }
}
\Return $\argmax_{x \in S} x$\;
\end{algorithm}

The privacy guarantee of this mechanism is stated next. Although \citealp{cohen2023generalized} proves a more general result, the statement below is sufficient for our purposes.
\begin{theorem}[\citealp{cohen2023generalized}]\label{thm:generalized_select}
\label{thm:cohen_selection}
Suppose each mechanism $\cM_i$ is $(\eps,\delta_i)$-DP with $\delta_i \ge 0$. Then, for any choice of $\tau \in \mathbb{N}$ and $\gamma \in (0,\infty)$, the algorithm $\Selection_{\tau,\gamma}$ is $((2+\gamma)\eps, \tau \sum_{i=1}^k \delta_i)$-DP.
\end{theorem}

\section{Proofs from \Cref{sec:monotone}}\label{appendix:monotone} In this section, we provide the missing proofs from \Cref{sec:monotone} for monotone submodular maximization. We begin with the proof of \Cref{thm:two_guess_greedy_intro}. Then, we move on to present our faster algorithm in this setting and prove \Cref{thm:greedyplus_intro}.

\subsection{Proof of \Cref{thm:two_guess_greedy_intro}}

\begin{lemma}[\Cref{lemma:two_guess_greedy_complexity} restated]
    \Cref{alg:two_guess_greedy} makes $O(\beta^{-1}n^3k)$ oracle queries.
\end{lemma}

\begin{proof}
Each execution of \Cref{algorithm:greedy} consists of at most $k$ iterations of the \texttt{while} loop, as the cardinality of the solution increases by one and remains feasible after each iteration. Each iteration evaluates $O(n)$ marginal gains, resulting in $O(nk)$ oracle queries per execution. \Cref{alg:two_guess_greedy} invokes \Cref{algorithm:greedy} at most $O(\beta^{-1})$ times for every subset of $\cN$ with cardinality at most $2$, totaling $O(n^2 \beta^{-1})$ invocations. Consequently, the total query complexity is $O(\beta^{-1}n^3k)$.
\end{proof}

\subsubsection{Privacy Analysis}
\begin{lemma}[\Cref{lemma:2guess_privacy} restated]\label[lemma]{lemma:2guess_privacy_formal}
For any $\eps,\delta\in (0,1]$, \Cref{alg:two_guess_greedy} is $(\eps,\delta)$-DP.
\end{lemma}
\begin{proof}
Let $\delta' = \frac{\delta}{n^2 \lceil 3/\beta \rceil}$. We begin by analyzing the privacy guarantee of \Cref{algorithm:greedy}, instantiated with privacy parameter $\eps_0/6$, where
\[
    \eps_0 = \frac{\eps}{2\sqrt{2k\log(1/\delta')}}.
\]
The algorithm invokes \GenRNMT sequentially at most $k$ times, each with parameter $\eps_0/6$. By the privacy guarantee of \GenRNMT (\Cref{thm:GRNM_formal}) and the advanced composition theorem (\Cref{thm:composition}), it follows that \Cref{algorithm:greedy} is $(\eps/6,\delta')$-DP.

We next consider any subset $Y \subseteq \cN$ of size at most two and define a corresponding mechanism $\cM_Y$. The mechanism first samples noise $z_Y \sim \mathrm{Lap}(\tfrac{6\Delta_f}{\eps})$. If $|Y|=2$, it runs \Cref{algorithm:greedy} on the residual instance induced by $Y$, obtains its output $S_Y$, and outputs the pair $(S_Y, f(S_Y \cup Y) + z_Y)$. By the guarantee of the Laplace mechanism (\Cref{thm:laplace_guarantees}), the privacy guarantee of \Cref{algorithm:greedy} established above, and basic composition (\Cref{thm:composition}), $\cM_Y$ is $(\eps/3,\delta')$-DP. If $|Y|<2$, the mechanism outputs $(Y, f(Y)+z_Y)$, which is $(\eps/6,0)$-DP and in particular $(\eps/3,\delta')$-DP.

Finally, \Cref{alg:two_guess_greedy} can be viewed as an instantiation of $\Selection_{\tau,\gamma}$ applied to the collection of mechanisms
\[
\{\cM_Y \mid Y \subseteq \cN,\ |Y| \le 2,\ c(Y) \le B\},
\]
with parameters $\tau = \lceil 3/\beta \rceil$ and $\gamma = 1$. By the privacy guarantees of the generalized private selection procedure (\Cref{thm:generalized_select}), the overall algorithm satisfies $(\eps,\delta)$-DP, since
\[
    \tau \cdot \sum_{Y} \delta' \le \tau n^2 \delta' = \delta.
\]
\end{proof}

\subsubsection{Utility Analysis}
We follow the proof outlined by \citealp{feldman2023practical}. We begin by recalling the notation introduced in \Cref{sec:monotone}. Let $\OPT$ denote an optimal solution, and assume first that $|\OPT| > 2$. Let $Y^* \subseteq \OPT$ be a subset of size two with maximum value, and consider the residual instance defined by the ground set $\cN \setminus Y^*$, budget $1 - c(Y^*)$, and objective function $h(S) = f(S \mid Y^*)$. Note that $h$ is a non-negative, monotone, and submodular function. Recall that we normalize the total budget to $B=1$.

Let $O = \OPT \setminus Y^*$, and let $o_m$ be an element in $O$ with maximum cost, i.e., $o_m \in \arg\max_{o \in O} c(o)$. Note that $O$ is feasible in the residual instance and is non-empty by the assumption that $|\OPT|>2$. Let $S_1, \dots, S_\ell$ denote the sequence of partial solutions generated by \Cref{algorithm:greedy}, where $S_\ell$ is the returned solution. 
Define:
\[
T = \min \left\{ i \;\mid\; (1 \le i \le \ell - 1 \text{ and } c(S_i) > 1 - c(Y^*) - c(o_m)) \text{ or } (i = \ell) \right\}.
\]
Intuitively, $T$ is the first index among $i=1,\dots,\ell-1$ such that the cost exceeds $1 - c(Y^*) - c(o_m)$, if such an index exists; otherwise, $T = \ell$.

\begin{lemma}[\Cref{lemma:sensitivity_bound} restated]\label{lemma:sensitivity_bound_formal}
For every iteration $i$ of \Cref{algorithm:greedy} and $u \in \mathcal{C}_i$, the quality score $q_i(u;D)$ has sensitivity $2\Delta_f/ c(u)$.
\end{lemma}
\begin{proof}
    Recall that for all $u\in \cC_i$ we have $q_i(u;D)=\frac{h_D(u\mid S_i)}{c(u)}$, and $h_D(S)=f_D(S\mid Y)$ for some subset $Y\subseteq \cN$. Expanding the definition of $h_D$, we obtain 
    \[
        q_i(u;D)= \frac{h_D(S_i\cup \set{u}\mid Y)-h_D(S_i\mid Y)}{c(u)} = \frac{f_D(S_i\cup Y\cup \set{u})-f_D(S_i\cup Y)}{c(u)}.
    \]
       Since $f_D(S)$ has sensitivity at most $\Delta_f$ for any $S\subseteq \cN$, for any two neighboring datasets $D,D'$
\begin{align*}
|q_i(u;D) - q_i(u;D')|
&\le \frac{1}{c(u)} \bigl| f_D(S_i \cup Y \cup \{u\}) - f_{D'}(S_i \cup Y \cup \{u\}) \bigr| \\
&\quad + \frac{1}{c(u)} \bigl| f_{D'}(S_i \cup Y) - f_D(S_i \cup Y) \bigr| \\
&\le \frac{2\Delta_f}{c(u)}.
\end{align*}
\end{proof}

\begin{lemma}[Formal version of \Cref{lemma:gain_progress}]\label[lemma]{lemma:gain_progress_formal}
Let $Y^* \subseteq \OPT$ be a maximum-value subset of size $2$. Suppose \Cref{algorithm:greedy} is instantiated with parameters $\eps$ and $\beta$ and executed on the residual instance defined by $Y^*$, i.e., with ground set $\cN \setminus Y^*$, budget $1 - c(Y^*)$, and objective $h(S) = f(S \mid Y^*)$. Then, with probability $1-\beta$, for every $i=0,\dots, T-1$:
    \[
        (1-c(Y^*)-c(o_m))\cdot \paren{\frac{h(u_{i+1}\mid S_i)}{c(u_{i+1})}} \ge h(O) - h(S_i) - h(o_m) - \frac{8k\Delta_f}{\eps}\log\frac{nk}{\beta}.
    \]
\end{lemma}

\begin{proof}
Since the sets $S_i$ remain feasible in all iterations, the number of iterations is at most $k$. By the utility guarantee of \GenRNMT (\Cref{thm:GRNM_formal}) and the sensitivity bound (\Cref{lemma:sensitivity_bound_formal}), in each iteration $i$, with probability at least $1 - \beta/k$, \GenRNMT selects an element $u_{i+1}$ satisfying:
\begin{align}
    \label{eq:GEM_guarantee1}
    \frac{h(u_{i+1}\mid S_i)}{c(u_{i+1})} \ge \max_{u \in \cC_i} \left\{ \frac{h(u \mid S_i)}{c(u)} - \frac{8\Delta_f}{\eps}\log\frac{nk}{\beta} \right\}.
\end{align}
By a union bound, this holds for all iterations with probability at least $1 - \beta$. We condition on this event for the remainder of the proof.

Let $R_i = S_i \cup \{o_m\}$ and let $E = \frac{8\Delta_f}{\eps}\log\frac{nk}{\beta}$. We observe:
\begin{align}
    h(O) - h(S_i \cup \{o_m\}) &= h(O) - h(R_i) \nonumber\\
    &\le h(O \cup R_i) - h(R_i) \tag{Monotonicity} \nonumber\\
    &\le \sum_{s \in O \setminus R_i} h(s \mid R_i) \tag{Submodularity} \nonumber\\
    &= \sum_{s \in O \setminus R_i} c(s) \frac{h(s \mid R_i)}{c(s)} \nonumber\\
    &\le \sum_{s \in O \setminus R_i} c(s) \cdot \left( \frac{h(u_{i+1} \mid S_i)}{c(u_{i+1})} + \frac{E}{c(s)} \right) \label{eq22} \\
    &= c(O \setminus R_i) \cdot \frac{h(u_{i+1} \mid S_i)}{c(u_{i+1})} + |O \setminus R_i| \cdot E\nonumber \\
    &\le (c(O) - c(O \cap R_i)) \cdot \frac{h(u_{i+1} \mid S_i)}{c(u_{i+1})} + kE \nonumber\\
    &\le (1 - c(Y^*) - c(o_m)) \cdot \frac{h(u_{i+1} \mid S_i)}{c(u_{i+1})} + kE.
\end{align}
For inequality \eqref{eq22}, note that since $i < T$, we have $c(S_i) + c(o_m) \le 1 - c(Y^*)$. Thus, every $s \in O \setminus R_i$ can be added to $S_i$ without violating feasibility (i.e., $O \setminus R_i \subseteq \cC_i$), and the bound follows from \eqref{eq:GEM_guarantee1} and the fact that $h(s \mid R_i) \le h(s \mid S_i)$ by submodularity. The final inequality holds by $c(O) \le 1 - c(Y^*)$ and $o_m \in O \cap R_i$.

By the submodularity and non-negativity of $h$ we have $h(S_i \cup \{o_m\}) \le h(S_i) + h(o_m)$. Rearranging the terms in the established chain of inequalities then yields the lemma.
\end{proof}

\begin{lemma}\label[lemma]{lemma:greedy_performance} 
Suppose \Cref{algorithm:greedy} is instantiated with parameters $\eps$ and $\beta$, and executed on the residual instance defined as in \Cref{lemma:gain_progress_formal}. Then, with probability $1-\beta$,
Then, with probability $1-\beta$,
\[
h(S_\ell) \ge (1 - \tfrac{1}{e})\cdot  \bigl( h(O) - h(o_m) \bigr)
- \frac{8k\Delta_f }{\eps} \log \frac{nk}{\beta}.
\]
\end{lemma}

\begin{proof}
Let us denote  $E = \frac{8\Delta_f}{\eps}\log\frac{nk}{\beta}$.
We may assume without loss of generality that $h(S_i)< h(O)-h(o_m)-kE$ for every $0\le i\le T$, since otherwise the lemma follows immediately from the monotonicity of $h$.
Applying \Cref{lemma:gain_progress} for every $0\le i\le T-1$, we get
 \begin{align*}
      \frac{h(S_{i+1})-h(S_{i})}{c(u_{i+1})}=\frac{h(u_{i+1}\mid S_i)}{c(u_{i+1})}\ge \frac{h(O)- h(S_i)-h(o_m) -kE}{1 -c(Y^*) -c(o_m)} && \forall i\in \set{0,\dots,  T-1}.
 \end{align*}
Rearranging yields
\begin{align*}
    h(O)-h(S_{i+1})-h(o_m) -kE &\le \paren{1-\frac{c(u_{i+1})}{1 -c(Y^*) -c(o_m)}}[h(O)-h(S_i)-h(o_m)-kE]\\
    \le &e^{-c(u_{i+1})/(1 -c(Y^*) -c(o_m))}\cdot[h(O)-h(S_i)-h(o_m)-kE]
\end{align*}
Unraveling the last inequality for all $0\le i\le T-1$ gives
\begin{align*}
 h(O)-h(S_T)-h(o_m)-kE  \le & \prod_{i=0}^{T-1} e^{-c(u_{i+1})/(1 -c(Y^*) -c(o_m))}\cdot[h(O)-h(S_0)-h(o_m)-kE]\\
   &=  e^{-c(S_{T})/(1 -c(Y^*) -c(o_m))}\cdot[h(O)-h(S_0)-h(o_m)-kE]\\
    &\le  e^{-1}\cdot[h(O)-h(o_m)-kE]
\end{align*}
To see why the last inequality holds, consider two cases. If $T<\ell$, then by definition of $T$ we have $c(S_T)>1 -c(Y^*) -c(o_m)$. Now consider the case $T=\ell$. If $O\subseteq S_\ell$, then the inequality in the lemma holds by monotonicity of $h$. Otherwise, there exists an element $o\in O\setminus S_\ell$. Since \Cref{algorithm:greedy} has terminated at iteration $\ell$, no more elements can be added while maintaining feasibility in the residual instance, and therefore $c(S_\ell)>1 -c(Y^*)-c(o)\ge1 -c(Y^*) -c(o_m)$.

By rearranging we get
\[
    h(S_T) \ge (1 - \tfrac{1}{e})[ h(O)-h(o_m) -kE]
\]
The lemma follows by monotonicity.

\end{proof}

\begin{lemma}[Formal version of \Cref{lemma:performance_given_Y}]
Suppose \Cref{algorithm:greedy} is instantiated with parameters $\eps>0$ and $\beta\in(0,1)$, and executed on the residual instance defined as in \Cref{lemma:gain_progress_formal}. Then, with probability $1-\beta$,

\[
    f(Y^*\cup S_\ell) \ge (1 - \tfrac{1}{e})\cdot f(\OPT) - \frac{8k\Delta_f}{\eps} \log \frac{nk}{\beta}.
\]
\end{lemma}

\begin{proof}
With probability $1-\beta$, the event stated in \Cref{lemma:gain_progress} holds. Let $S_1,\dots, S_\ell$ be a sequence of sets computed by the \Cref{algorithm:greedy} conditioned on this event.
We have
\begin{align*}
    f(Y^*\cup S_\ell)=h(S_\ell)+f(Y^*) &\ge (1-\tfrac{1}{e})\cdot  \bigl[ h(\OPT\setminus Y^*)-h(o_m)\bigr] +f(Y^*) -\frac{8k\Delta_f}{\eps} \log \frac{nk}{\beta} \\
    &=(1-\tfrac{1}{e})\cdot f(\OPT) + \tfrac{1}{e}\cdot  f(Y^*) -(1-\tfrac{1}{e}) f(o_m\mid Y^*)-\frac{8k\Delta_f}{\eps} \log \frac{nk}{\beta},
\end{align*}
where the first inequality holds by \Cref{lemma:greedy_performance}
Hence it remains to show that $\tfrac{1}{e}f(Y^*) \ge  (1-\tfrac{1}{e}) f(o_m \mid Y^*)$. The proof is identical to that in \cite{feldman2023practical}, and follows from the definition of $Y^*$ and submodularity:
\begin{align*}
\tfrac{1}{e}\cdot f(Y^*) &\ge \tfrac{1}{e}\cdot \bigl[ f(\{u_1,o_m\}) + f(\{u_2,o_m\}) - f(Y^*) \bigr] \\
&= \tfrac{1}{e}\cdot \bigl[ f(o_m \mid \{u_1\}) + f(o_m \mid \{u_2\}) + f(u_1) + f(u_2) - f(Y^*) \bigr] \\
&\ge \tfrac{2}{e}\cdot f(o_m \mid Y^*) \\
&\ge (1-\tfrac{1}{e}) \cdot f(o_m \mid Y^*) .
\end{align*}

\end{proof}

The next lemma concludes the utility proof for \Cref{alg:two_guess_greedy}.We remove the assumption that $|\OPT| > 2$ and establish the guarantee for the general case.
\begin{lemma}[\Cref{lemma:two_guess_greedy_utility} restated]
For any $\eps, \delta>0$, $ \beta\in(0,1)$, \Cref{alg:two_guess_greedy} outputs a feasible $S^*$ such that, with probability $1-\beta$,  
\[
f(S)\ge (1-\tfrac{1}{e})\cdot f(\OPT)
-O\paren{
\frac{k^{1.5}\Delta_f}{\eps}
\sqrt{\log\tfrac{n}{\delta\beta}}
 \log\tfrac{n}{\beta}}.
\]
\end{lemma}
\label{lemma:2guess_utility_formal}
\begin{proof}
We first consider the case $\abs{\OPT} \ge 2$ and focus on the iteration of \Cref{alg:two_guess_greedy} in which the algorithm guesses $Y^*$. In this iteration, the number of executions of \Cref{algorithm:greedy} is distributed uniformly over $\{0,\dots,\lceil 3/\beta \rceil\}$. Indeed, since $p$ is sampled uniformly from $[0,1]$ and the number of executions is distributed as $\Bin{\lceil 3/\beta \rceil}{p}$, each outcome occurs with equal probability.\footnote{For all $m \in \{0,1,\dots,\tau\}$, $\mathbb{E}_{p \sim U[0,1]} \binom{\tau}{m} p^m (1-p)^{\tau-m} = 1/(\tau+1)$ \cite{cohen2023generalized}.}
Consequently, \Cref{algorithm:greedy} is executed at least once with probability at least $1-\beta/3$. We condition on this event.

If $\abs{\OPT}>2$, then by \Cref{lemma:performance_given_Y}, with probability at least $1-\beta/3$, the first execution of \Cref{algorithm:greedy} in this iteration outputs a set $S_{Y^*,j}$ satisfying
\begin{align}
    f(Y^* \cup S_{Y^*,j}) \ge (1-\tfrac{1}{e}) \cdot f(\OPT) - \frac{48k\Delta_f}{\eps_0} \log \frac{3nk}{\beta}.
    \label{eq:lowerbound_YSYj}
\end{align}
If instead $\abs{\OPT}=2$, then $Y^*=\OPT$, and the same inequality follows immediately from the monotonicity of $f$. We therefore condition on the event that \eqref{eq:lowerbound_YSYj} holds.

Next, standard Laplace tail bounds and a union bound imply that, with probability at least $1-\beta/3$, every Laplace noise  sampled in Line~\ref{alg2:line_sample_laplace} during the execution of \Cref{alg:two_guess_greedy} has magnitude at most $\frac{6\Delta_f}{\eps}\log\frac{3n}{\beta}$.
Conditioning on this event, the set $S^*$ selected in Line~\ref{alg2:get_best} satisfies
\[
    f(S^*) \ge f(Y^* \cup S_{Y^*,j}) - \frac{12\Delta_f}{\eps}\log\frac{3n}{\beta}.
\]

Combining the above events and applying a union bound, we conclude that with probability at least $1-\beta$,
\begin{align*}
    f(S^*) 
    &\ge (1-\tfrac{1}{e}) \cdot f(\OPT)
    - \frac{12\Delta_f}{\eps}\log\frac{3n}{\beta}
    - \frac{48k\Delta_f}{\eps_0} \log \frac{3nk}{\beta} \\
    &\ge (1-\tfrac{1}{e}) \cdot f(\OPT)
    - O\!\left(\frac{k\Delta_f}{\eps_0} \log \frac{n}{\beta}\right),
\end{align*}
where we have used that $k \le n$. The inequality in the lemma follows by substitution the value for $\eps_0$ defined in the proof of \Cref{lemma:2guess_privacy_formal}. In particular,  $\eps_0 = O\Big(\eps / \sqrt{k \log\frac{n}{\delta\beta}}\Big)$.

It remains to consider the case $\abs{\OPT}=1$. In this case, one iteration of the algorithm guesses $Y=\OPT$. As above, with probability at least $1-\beta/3$, all Laplace noises sampled in Line~\ref{alg2:line_sample_laplace} are bounded by $\frac{6\Delta_f}{\eps}\log\frac{3n}{\beta}$, and the candidate set considered in Line~\ref{alg2:get_best} includes $\OPT$. Conditioning on this event, the output $S^*$ therefore satisfies
\[
    f(S^*) \ge f(\OPT) - \frac{12\Delta_f}{\eps}\log\frac{3n}{\beta},
\]
which completes the proof.
\end{proof}

\subsubsection{Improved Complexity Dependence on $\beta$}

In \Cref{section:challenges_results_monotone}, we noted that the query complexity in \Cref{thm:two_guess_greedy_intro} can be reduced to $O(nk^3 \log(1/\beta))$ at the cost of an additional $\log(1/\beta)$ factor in the additive error. We provide a proof sketch of this modification below

Recall that in the standard configuration of \Cref{alg:two_guess_greedy}, the threshold $p \in [0,1]$ is sampled uniformly, and each partial enumeration step is repeated $\lceil 3/\beta \rceil$ times. This corresponds to an instantiation of the selection mechanism $\Selection_{\tau,\gamma}$ with $\tau=\lceil 3/\beta \rceil$ and $\gamma=1$, applied to the collection of mechanisms $\{\cM_Y \mid Y \subseteq \cN,\ |Y|\le 2,\ c(Y)\le B\}$. To achieve the improved query complexity, we instead sample $p \in [0,1]$ according to the cumulative distribution function $\Pr[p \le x] = x^{\log(6/\beta)}$, corresponding to $\gamma=\log(6/\beta)$, and set the repetition threshold to $\tau=\lceil e/\log(6/\beta)\rceil$. This leads to a query complexity of $O(nk^3 \log(1/\beta))$.

To analyze the resulting utility, observe that $\Pr[p \le 1/e] = (1/e)^{\log(6/\beta)} = \beta/6$. Thus, $p \ge 1/e$ holds with probability at least $1-\beta/6$. Conditioned on this event, the probability that $\cM_{Y^*}$ is executed at least once across $\tau$ trials is at least
\[
1 - (1 - 1/e)^\tau \ge 1 - \exp(-\tau/e) \ge 1 - \beta/6.
\]
By a union bound, $\cM_{Y^*}$ is executed at least once with probability at least $1-\beta/3$, allowing the downstream utility analysis of \Cref{lemma:two_guess_greedy_utility} to proceed unchanged. Finally, by \Cref{thm:generalized_select}, each invocation of \Cref{algorithm:greedy} now needs to use the tighter privacy parameter $\eps_0/(2(2+\gamma)) = O(\eps_0/\log(1/\beta))$ instead of $\eps_0/6$ that was used in the proof of \Cref{lemma:2guess_privacy_formal} (which corresponds to $\gamma=1$). Propagating this scaled parameter leads directly to the additional $O(\log(1/\beta))$ factor in the final additive error bound.

\subsection{Proof of \Cref{thm:greedyplus_intro}}\label{appendix:greedyplus}
We now present a more efficient DP algorithm for the monotone SMK problem. \Cref{algorithm:greedyplus} is based on the Greedy$^+$ algorithm of \cite{yaroslavtsev2020bring}, which executes density-greedy with the objective $f$ and returns the best solution among the partial solutions and all their feasible single-element extensions. \Cref{algorithm:greedyplus}  executes \Cref{algorithm:greedy}, and uses the \RNMT mechanism to select the best solution among the partial solutions, their feasible single-element extensions, and all singletons. This yields the guarantee stated in \Cref{thm:greedyplus_intro}, which is proven by combining \cref{lemma:greedyplus_complexity,lemma:greedyplus_privacy,lemma:greedyplus_utility}.

\begin{algorithm}
\caption{ \DPDGP (\texttt{DP-Density-Greedy$^+$})}
\label{algorithm:greedyplus}
\setcounter{AlgoLine}{0}
\KwIn{
Dataset $D$;
Submodular function $f_D:2^\cN\to\R_+$;
Budget $B$;
Parameters  $\beta,\eps,\delta\in (0,1]$.}

 $S_0\gets\varnothing$ and $i\gets 0$.\;

  $\eps_0 \gets \frac{\eps}{4\sqrt{2k \log(1/\delta)}}$ \plabel{greedyplus:seteps} 

Let $S_1,\dots,S_\ell$ be all the partial solutions obtained throughout the execution of \Cref{algorithm:greedy} on the instance defined by the ground set $\cN$, the objective function $h_D(S)=f_D(S)$, the budget $B$, the privacy parameter $\eps_0$, and utility parameter $\beta/2$. \plabel{alg:greedyplus_define_residual}

 Define  $\cS \gets \set{S_\ell} \cup \set{\set{u}\mid u\in \cN} \cup \set{S_{i}\cup\set{u} \mid 0\le i\le \ell,\ u \in \cN,\ c(S_{i}\cup\set{u})\le B } $, 
 
 Let $S^* \gets \RNM{\eps/2}{\cS, f}{D}$ \plabel{alg:greedyplus_lastEM}
 
\Return $S^*$
\end{algorithm}

\begin{lemma}\label[lemma]{lemma:greedyplus_complexity}
\Cref{algorithm:greedyplus} performs $O(nk)$ oracle queries.
\end{lemma}

\begin{proof}
\Cref{algorithm:greedy} executes at most $k$ iterations of its main loop, and in each iteration, it evaluates the marginal gain of at most $n$ elements, contributing $O(nk)$ oracle queries. 
The final invocation of \RNMT in Line~\ref{alg:greedyplus_lastEM} requires one oracle call for each candidate in the set $\cS$, which contains at most $nk+n+1$ candidates. 
Consequently, the total number of oracle queries made by \Cref{algorithm:greedyplus} is $O(nk)$.
\end{proof}

\begin{lemma}\label[lemma]{lemma:greedyplus_privacy}
For any $\eps,\delta\in (0,1]$, \Cref{algorithm:greedyplus} is $(\eps,\delta)$-DP.
\end{lemma}

\begin{proof}
The privacy guarantee follows from the composition of the iterative greedy phase and the final selection phase. First, we consider the execution of \Cref{algorithm:greedy} in Line~\ref{alg:greedyplus_define_residual}. By instantiating \Cref{algorithm:greedy} with the per-step privacy parameter $\eps_0 = \eps / (4\sqrt{2k \log(1/\delta)})$, advanced composition (\Cref{thm:composition}) implies that  releasing the sequence of partial solutions $S_1, \dots, S_\ell$ is $(\eps/2, \delta)$-DP.
Second, the final selection step via \RNMT is executed with privacy parameter $\eps/2$. According to \Cref{thm:RNM_formal}, this step is $(\eps/2, 0)$-DP. Combining this with the privacy guarantee of \Cref{algorithm:greedy} via basic composition yields that the overall algorithm is $(\eps,\delta)$-DP.
\end{proof}

\begin{lemma}\label[lemma]{lemma:greedyplus_utility}
For any $\eps, \delta>0$, $\beta\in (0,1)$, \Cref{algorithm:greedyplus} returns a feasible subset $S$ such that, with probability $1-\beta$, 
\[
f(S)\ge \tfrac{1}{2}\cdot f(\OPT)
-O\paren{
\frac{\Delta_f\, k^{1.5}}{\eps}
\sqrt{\log\tfrac{1}{\delta}}
 \log\tfrac{n}{\beta}}.
\]
\end{lemma}

\begin{proof}
We consider two cases based on the cardinality of $\OPT$ and the objective value of each single-element extension of a partial solution. Let $o_m \in \arg\max_{o \in \OPT} c(o)$ and let $T$ be the smallest index in $\{1, \dots, \ell-1\}$ satisfying $c(S_T) > 1 - c(o_m)$, or let $T = \ell$ if no such index exists.

\paratitle{Case 1} 
If $|\OPT| = 1$, then $\OPT \in \cS$ as all singletons are included in the candidate set. If for some $0\le i < T$, $f(S_i \cup \{o_m\}) \ge \frac{1}{2}f(\OPT)$, then there exists a set $S' \in \cS$ such that $f(S') \ge \frac{1}{2}f(\OPT)$. In both scenarios, the utility guarantee of \RNMT (\Cref{thm:genRNM_guarantee}) ensures that with probability at least $1-\beta$:
\[
    f(S^*) \ge \max_{S \in \cS} f(S) - \frac{4\Delta_f}{\eps} \log \frac{nk+n+1}{\beta} \ge \frac{1}{2}f(\OPT) - O\Big(\frac{\Delta_f}{\eps} \log \frac{n}{\beta}\Big).
\]

\paratitle{Case 2}
We now consider the case where $|\OPT| \ge 2$ and $f(S_i \cup \{o_m\}) < \frac{1}{2}f(\OPT)$ for all $i<T$.
Following a chain of inequalities analogous to that in the proof of \Cref{lemma:gain_progress_formal}, we obtain that with probability at least $1-\beta/2$,
\begin{align*}
     (1-c(o_m)) \cdot \frac{f(u_{i+1} \mid S_i)}{c(u_{i+1})}
     &\ge f(\OPT) - f(S_i \cup \{o_m\}) - kE \\
     &\ge \tfrac{1}{2} f(\OPT) - kE, && \forall i \in \{0, \dots, T-1\},
\end{align*}
where $E = \frac{8\Delta_f}{\eps_0} \log \frac{2nk}{\beta}$. Rearranging the inequality yields:
\begin{align*}
       \frac{c(u_{i+1})}{1-c(o_m)} \cdot \paren{\frac{f(\OPT)}{2} - kE} \le f(S_{i+1}) - f(S_i).
\end{align*}
We may assume $f(\OPT) \ge 2kE$, as otherwise the guarantee holds trivially. Summing over $0 \le i \le T-1$ yields:
\begin{align}
    f(S_\ell) \ge f(S_T) &\ge f(S_0) + \paren{\frac{f(\OPT)}{2} - kE} \cdot \sum_{i=0}^{T-1} \frac{c(u_{i+1})}{1-c(o_m)} \notag \\
    &\ge \paren{\frac{f(\OPT)}{2} - kE} \cdot \frac{c(S_T)}{1-c(o_m)} \ge \frac{f(\OPT)}{2} - kE, \label{eq:greedyplus}
\end{align}
where we use the identity $\sum_{i=0}^{T-1} c(u_{i+1}) = c(S_T)$ and the non-negativity of $f$. It therefore remains to show that $c(S_T) \ge 1 - c(o_m)$. If $T < \ell$, this inequality holds by the definition of $T$. Thus, assume that $T = \ell$. By monotonicity of $f$, it suffices the consider the case that $\OPT \not\subseteq S_\ell$. In this case, 
there exists an element $o \in \OPT \setminus S_\ell$. Since \Cref{algorithm:greedy} has terminated and $o$ was not added to the solution, it must be that $c(S_\ell) > 1 - c(o) \ge 1 - c(o_m)$.

Observe that since $S_\ell \in \cS$, \Cref{thm:genRNM_guarantee} implies that with probability at least $1-\beta/2$:
\[
    f(S^*) \ge f(S_\ell) - \frac{4\Delta_f}{\eps_0} \log \frac{nk+n+1}{\beta/2}.
\]

By a union bound and substituting 
$\eps_0 = \frac{\eps}{4\sqrt{2k \log(1/\delta)}}$, we conclude that with probability at least $1-\beta$:
\[
f(S^*) \ge \tfrac{1}{2} f(\OPT) - O\Bigg( \frac{\Delta_f\, k^{1.5}}{\varepsilon} \sqrt{\log \tfrac{1}{\delta}} \,\log \tfrac{n}{\beta} \Bigg).
\]
\end{proof}

\section{ Proofs from \Cref{sec:non-monotone}}
\label{sec:appendix_nonmonotone}
In this section, we provide the proofs deferred from \Cref{sec:non-monotone} on non-monotone submodular maximization, completing the proof of \Cref{thm:samplegreedy}.
For convenience in the utility analysis, we adopt a standard analytical setup where the DP noise random variables are realized at the beginning of execution. See \Cref{sec:appendix_nonmonotone_utility} for the formal treatment and definitions.

The query complexity of \Cref{algorithm:subsample_greedy} is given by the following lemma.

\begin{lemma}\label[lemma]{lemma:subsample_complexity_appendix}
 \Cref{algorithm:subsample_greedy} makes $O(nk)$ oracle queries.
\end{lemma}

\begin{proof}
Let $S_\ell = \set{u_1,\dots,u_{|S_\ell|}}$ be the solution obtained at the end of the \texttt{while} loop.
\Cref{algorithm:subsample_greedy} can be implemented to compute $O(n)$ marginal gains for each prefix $\set{u_1,\dots,u_j}$. Because $S_\ell$ is feasible, $|S_\ell|\le k$. Thus, the total number of oracle queries incurred within the \texttt{while} loop is $O(nk)$.
Finally, the last application of \RNMT\ in \Cref{alg:sample_lastEM} evaluates the objective function on each set in $\cS$, defined in Line~\ref{alg:sample_defineS}. 
Since we have at most $k$ distinct partial solutions, 
the candidate family $\cS$ contains at most $k + k\cdot n = O(nk)$ sets. Hence this step also requires $O(nk)$ oracle queries.
Overall, the total query complexity is $O(nk)$.
\end{proof}

\subsection{Privacy Analysis}
We now proceed to the privacy analysis. To this end, we utilize the following concentration bound for the sum of independent geometric random variables.
\begin{lemma}[\citealp{janson2018tail}]\label{lemma:janson}
Let $X_1,\dots,X_k$ be independent, geometrically distributed, random variables with 
$X_i \sim \mathrm{Geom}(p_i)$ for $0 < p_i \le 1$. Let $p_* = \min_i p_i$, and
 $\mu = \mathbb{E}[\sum_{i=1}^k X_i]$. 
Then, for every $\lambda \ge 1$,
\[
\Pr\Bigg[\sum_{j=1}^{k} X_j \ge \lambda \mu\Bigg] \le \exp\Big(- p_* \mu \, (\lambda - 1 - \log \lambda)\Big).
\]
\end{lemma}

We also rely on the following technical inequality.
\begin{lemma}\label[lemma]{lemma:technical_bound}
For all $x>0$, we have
\[
x - \log(1+x) \ge \frac{x^2}{2(1+x)}.
\]
\end{lemma}

\begin{proof}
Define $g(x) := x - \log(1+x) - \frac{x^2}{2(1+x)}$. Then
\[
g'(x) = 1 - \frac{1}{1+x} -\frac{x(2+x)}{2(1+x)^2} 
 = \frac{x^2}{2(1+x)^2} \ge 0.
\]
Since $g(0) = 0$ and $g$ is non-decreasing for $x>0$, we conclude $g(x)\ge 0$ for all $x>0$.
\end{proof}

As a consequence, we obtain the following lemma.
\begin{lemma}\label[lemma]{cor:boundSumGeom}
 Suppose $X_i \sim \mathrm{Geom}(1/2)$ for all $i$. Then, for every $\delta > 0$,
 \begin{align*}
 \Pr\Bigg[\sum_{j=1}^{k} X_j \ge (2+\sqrt{2})k+(4+\sqrt{2})\log(1/\delta)\Bigg] \le \delta.
 \end{align*}
\end{lemma}

\begin{proof}
 Let $X=\sum_{i=1}^k X_i$, and note that we have $p_* = 1/2$ and $\mu = \mathbb{E}[X]=2k$. For a parameter $t>0$ to be specified later, let $\lambda=1+\frac{t}{2k}$. Then by \Cref{lemma:janson},
 \begin{align}
 \Pr[X\ge 2k+t] &\le \exp\left(- k \left(\frac{t}{2k} - \log \left(1+\frac{t}{2k}\right)\right)\right) \nonumber\\
 &\le \exp\left(-k \cdot \frac{(t/2k)^2}{2(1+t/2k)}\right) = \exp\left(-\frac{t^2}{4(2k+t)}\right)
 \label{eq:prob_bound}
 \end{align}
 where the second inequality uses \Cref{lemma:technical_bound}, applied with $x = \frac{t}{2k}$. 

 In order for the above probability to be at most $\delta$, it suffices that
\begin{align*}
 \frac{t^2}{4(2k+t)} &\ge \log(1/\delta) \\
 \iff \quad
 t^2 - 4t \log(1/\delta) - 8k \log(1/\delta) &\ge 0 .
\end{align*}

Solving the quadratic inequality, we find that whenever 
\[
 t> 2\log(1/\delta) +\sqrt{4\log^2(1/\delta)+8k\log(1/\delta)}
\]
By \eqref{eq:prob_bound} we have $\Pr[X\ge 2k+t]\le \delta$. To further simplify the bound and complete the proof, observe that
\begin{align*}
\sqrt{2}k + (4+\sqrt{2})\log(1/\delta)
&\;\ge\; 4\log(1/\delta) + 2\sqrt{2k\log(1/\delta)} \\[4pt]
&=\; 2\log(1/\delta)
 + 2\log(1/\delta)
 + 2\sqrt{2k\log(1/\delta)} \\[4pt]
&\;\ge\; 2\log(1/\delta)
 + 2\log(1/\delta)\Bigl(1 + \sqrt{\tfrac{2k}{\log(1/\delta)}}\Bigr) \\[4pt]
&\;\ge\; 2\log(1/\delta)
 + 2\log(1/\delta)\sqrt{1 + \tfrac{2k}{\log(1/\delta)}} \\[4pt]
&=\; 2\log(1/\delta)
 + \sqrt{4\log^2(1/\delta) + 8k\log(1/\delta)},
\end{align*}
where the first inequality follows from the AM-GM inequality, which implies that $\sqrt{2}k + \sqrt{2}\log(1/\delta) \ge 2\sqrt{2k\log(1/\delta)}$.
\end{proof}

\begin{lemma}\label[lemma]{lemma:runtime_bound_formal}
 \Cref{algorithm:subsample_greedy} executes the \texttt{while} loop (Lines~\ref{alg:sample_line2}-\ref{alg:sample_last_inloop}) at most $ (2+\sqrt{2})k+(4+\sqrt{2})\log(2/\delta)$ times with probability at least $1-\delta/2$.
\end{lemma}

\begin{proof}
For the purpose of analysis, we assume without loss of generality that if the loop terminates with $|S_\ell| < k$, the algorithm sets $S_{\ell+1}, \dots, S_k \gets S_\ell$. This postprocessing does not affect the privacy analysis and ensures the algorithm executes exactly $k$ phases.

For $j = 1, \dots, k$, let $N_j$ denote the number of iterations of the \texttt{while} loop during which the algorithm has selected exactly $j-1$ elements. Formally, if $j \le \ell$, then $N_j$ is the number of iterations from the addition of the $(j-1)$-th element up to and including the addition of the $j$-th element. If $j > \ell$, we set $N_j = 0$.
The algorithm remains in phase $j$ only if the random inclusion step in Line~\ref{alg:sample_random_inclusion} rejects the selected element. Since the probability of acceptance in each trial is $1/2$ and independent of the prior state, the number of trials until an element is accepted follows a geometric distribution with parameter $1/2$. 
We couple the execution with a sequence of $k$ independent random variables $X_1, \dots, X_k \sim \mathrm{Geom}(1/2)$ ensuring that $N_j \le X_j$ for all $j$. For each phase $j \le \ell$, we have $N_j = X_j$, representing the number of trials until an acceptance. For $j > \ell$, we have $N_j = 0 \le X_j$. This coupling implies that the total number of iterations $\sum_{j=1}^{k} N_j$ is stochastically dominated by $\sum_{j=1}^{k} X_j$. That is, for all $t \ge 0$:
\[
\Pr\left[\sum_{j=1}^{k} N_j > t\right] \le \Pr\left[\sum_{j=1}^{k} X_j > t\right].
\]
The lemma follows by applying \Cref{cor:boundSumGeom} to the right-hand side with $\delta' = \delta/2$.
\end{proof}
We now prove the privacy guarantee of \Cref{algorithm:subsample_greedy}.

\begin{lemma}\label[lemma]{lemma:sample_privacy_proof}
For any $\eps,\delta>0$, \Cref{algorithm:subsample_greedy} is $(\eps,\delta)$-DP.
\end{lemma}
\begin{proof}
Define $L = (2+\sqrt{2})k + (4+\sqrt{2}) \log(2/\delta)$. Each iteration of the \texttt{while} loop (Lines~\ref{alg:sample_line2}--\ref{alg:sample_last_inloop}) applies the \GenRNMT with privacy parameter $\eps_0$, where
\[
\eps_0 = \tfrac{\eps}{4 \sqrt{2 L \log(2/\delta)}}.
\]
Furthermore, each quality score $q_i(u;D)$ has sensitivity $2\Delta_f/c(u)$, by an argument similar to that in the proof of \Cref{lemma:sensitivity_bound_formal}, and $f_D(u_{i+1}\mid S_i)$ has sensitivity $2\Delta_f$. Thus, by \Cref{thm:genRNM_guarantee} each single iteration, which releases a candidate element $u_{i+1}$, is $\eps_0$-DP. 
Advanced composition (\Cref{thm:composition}) implies that releasing the outputs of the first $L$ iterations is $(\eps/2, \delta/2)$-DP. By \Cref{lemma:runtime_bound_formal}, the loop exceeds $L$ iterations with probability at most $\delta/2$. A standard union bound over these two failure events implies that the algorithm up to Line~\ref{alg:sample_defineS} is $(\eps/2, \delta)$-DP. 
Finally, by basic composition with the final \RNMT step (Line~\ref{alg:sample_lastEM}), which uses privacy parameter $\eps/2$, the overall algorithm is $(\eps,\delta)$-DP.
\end{proof}

\subsection{Utility Analysis}\label{sec:appendix_nonmonotone_utility}
To simplify the analysis of \Cref{algorithm:subsample_greedy}, we assume, without loss of generality, that all differential privacy noise random variables are sampled at the beginning of the algorithm. As noted in \Cref{remark;GRNM}, for each application of \GenRNMT within \Cref{algorithm:subsample_greedy}, we use $n$ as a uniform upper bound for the candidate set size instead of the iteration-specific $|\cC_i|$. Consequently, we set $t = \tfrac{2}{\eps_0}\log\tfrac{2n^2}{\beta}$. This allows us to isolate the randomness of the sampling step in Line~\ref{alg:sample_random_inclusion}. Specifically, we assume the following random variables are sampled upfront:
\begin{itemize}
 \item $n^2$ independent samples from $\Expo{\frac{\eps_0}{2}}$ for the \GenRNMT calls in Line~\ref{alg:sample_GEM}.
 \item $nk+k$ independent samples from $\Expo{\frac{\eps}{4}}$ for the \RNMT call in Line~\ref{alg:sample_lastEM}.
\end{itemize}
The proof of the utility guarantee for \Cref{algorithm:subsample_greedy} follows the proof outlined by \citet{han2021approximation}. We extend their techniques to account for the additive errors introduced by the DP mechanisms. We begin by introducing events under which these errors are bounded.
\begin{itemize}
 \item \textbf{Event $\cE_1$:}\; All $n^2$ samples from $\Expo{\frac{\eps_0}{2}}$ are bounded by $\frac{2}{\eps_0}\log \frac{2n^2}{\beta}$

 \item \textbf{Event $\cE_2$:}\; All $nk+k$ samples from $\Expo{\frac{\eps}{4}}$ 
 are bounded by 
 $\frac{4\Delta_f}{\eps}\log \frac{2k(n+1)}{\beta}$.
 \item \textbf{Event $\cE$:}\; $\cE=\cE_1\cap\cE_2$.
\end{itemize}

\begin{lemma}\label[lemma]{lemma:good_event}
Conditioned on $\cE$, the following hold.
\begin{enumerate}[label=(\textit{\roman*})]
\item
Each element $u_{i+1}$ selected in Line~\ref{alg:sample_GEM} of \Cref{algorithm:subsample_greedy} satisfies
\begin{align}
\label{eq:GEM_guarantee}
q_i(u_{i+1};D)
=
\frac{f(u_{i+1}\mid S_i)}{c(u_{i+1})}
\ge
\max_{u\in \cC_i}
\left\{
\frac{f(u\mid S_i)}{c(u)}
-
\frac{8\Delta_f}{c(u)\eps_0}\log\frac{2n^2}{\beta}
\right\}.
\end{align}

\item
\Cref{algorithm:subsample_greedy} outputs a set $S^*$ satisfying
\[
f(S^*)
\ge
\max_{S\in \cS} f(S)
-
\frac{4\Delta_f}{\eps}\log \frac{2k(n+1)}{\beta}.
\]
\end{enumerate}
\end{lemma}

\begin{proof}
We begin by proving \textit{(i)}. In each iteration $i$, \GenRNMT\ is invoked with privacy parameter $\eps_0$. The quality function $q_i(u;D)$ has a sensitivity of $2\Delta_f / c(u)$ with respect to the dataset argument. Under event $\cE$, the $n^2$ exponential noise samples utilized across all \GenRNMT\ invocations are bounded by $\frac{2}{\eps_0}\log \frac{2n^2}{\beta}$. Adapting the utility analysis of \GenRNMT\ (\Cref{thm:GRNM_formal}) to this noise bound, we obtain
\[
q_i(u_{i+1};D) \ge \max_{u\in\cC_i}
\left\{
q_i(u;D) - \frac{4\Delta_{q_i(u;\cdot)}}{\eps_0} \log\frac{2n^2}{\beta}
\right\} = \max_{u\in\cC_i}
\left\{
q_i(u;D) - \frac{8\Delta_f}{c(u)\eps_0} \log\frac{2n^2}{\beta}
\right\}.
\]

For \textit{(ii)}, we consider the final \RNMT\ step in Line~\ref{alg:sample_lastEM}. Under event $\cE$, all $n+nk$ realized exponential noises are bounded by $\frac{4\Delta_f}{\eps}\log \frac{2k(n+1)}{\beta}$. Because \RNMT\ selects the element with the maximum noisy score, by \Cref{thm:RNM_formal}, the true value of the selected element differs from the true maximum value by at most the bound of the noise. This yields the stated bound.
\end{proof}

\begin{lemma}\label[lemma]{lemma:small_noise_event}
The event $\cE$ occurs with probability at least $1-\beta$.
\end{lemma}
\begin{proof}
The lemma follows from applying tail bounds to each noise distribution followed by a union bound.
\begin{itemize}
 \item Event $\cE_1$: For $X \sim \text{Exp}(\frac{\eps_0}{2})$, we have $\Pr[X > s] = e^{-\eps_0 s / 2}$. Setting $s = \frac{2}{\eps_0}\log \frac{2n^2}{\beta}$ gives $\Pr[X > s] = \frac{\beta}{2n^2}$. By a union bound over $n^2$ samples, $\Pr[\cE_1^c] \le \beta/2$.
 
 \item Event $\cE_2$: For $Z \sim \text{Exp}(\frac{\eps}{4\Delta_f})$, setting $s = \frac{4\Delta_f}{\eps}\log \frac{2k(n+1)}{\beta}$ gives $\Pr[Z > s] = \frac{\beta}{2k(n+1)}$. A union bound over $k(n+1)$ samples yields $\Pr[\cE_2^c] \le \beta/2$.
\end{itemize}
Combining these via a union bound over the complements $\cE_1^c$, and $\cE_2^c$ gives $\Pr[\cE] \ge 1-\beta$.
\end{proof}

We now move on to introduce notation extended from \cite{han2021approximation}.
Let $\ell$ be the last iteration of the while loop in \Cref{algorithm:subsample_greedy}, and let $u_1,\dots,u_\ell$ be the elements found in Line \ref{alg:sample_GEM}. Define $\tau(u_j)=j$ for $1\le j\le \ell$ and $\tau(v)=\infty$ for any $v\in V\setminus \set{u_1,\dots,u_\ell}$.
Throughout the proof, we denote
$
E \eqdef \tfrac{8\Delta_f}{\eps_0}\log \tfrac{2n^2}{\beta}
$
which will serve as a bound for the error terms.
For any realization of the randomness of the algorithm, define
\begin{align*}
&r = \min \Bigg\{ i \;\mid\; \Bigg(\max_{\substack{ v \in \OPT \setminus S_i \\c(S_i \cup \{v\}) \le B }} f(v\mid S_i) \le E\Bigg) \lor (i = \ell )\Bigg\}\\
&T = \min \left\{ i \;\mid\; \Big( 0 \le i \le r - 1 \wedge c(S_i) + c(u_{i+1}) > c(\OPT\setminus \{o_m\})\Big) \vee (i = r) \right\}, \\
&O_{\le T} = \left\{ v \in \OPT \;\mid\; \tau(v) \le T \right\}, \\
&O_{> T} = \left\{ v \in \OPT \;\mid\; \tau(v) > T \wedge f(v \mid S_T) > E \right\}.
\end{align*}
Intuitively, $r$ is the smallest index in $\{0, \dots, \ell\}$ such that no remaining element $v \in \OPT \setminus S_i$ whose addition remains feasible provides a marginal gain exceeding $E$, if such an index exists. Otherwise, $r=\ell$. Moreover, $T$ is the smallest index in $\{0, \dots, r-1\}$ such that adding $u_{T+1}$ into $S_T$ would make the cost exceed $c(\OPT \setminus \{o_m\})$, if such an index exists. Otherwise, $T = r$, and we have $c(S_r) \le c(\OPT \setminus \{o_m\})$.

The following lemma from \cite{han2021approximation}, constructs a fractional allocation of the cost of late optimal items onto elements of the algorithm’s partial solution, which allows the proof to upper-bound the total marginal contribution of those optimal items in terms of elements already selected by the algorithm. Note that our definition of $O_{>T}$ includes a stricter condition on the marginal gain than the one used in~\cite{han2021approximation}. Thus, the cost inequality $c(O_{>T} \setminus \{o_m\}) < c([S_T \setminus (\OPT \setminus \{o_m\})] \cup \{u_{T+1}\})$ required for the existence of $\Psi$ remains satisfied.
\begin{lemma}[Lemma 1 in \citealp{han2021approximation}]\label{lemma:mapping}
For any realization of the randomness in \Cref{algorithm:subsample_greedy} such that if $T<\ell$, there exists a mapping $\Psi$ satisfying the following properties. For each 
$u \in O_{>T} \setminus \set{o_m}$, $\Psi(u)$ is a set of 2-tuples such that each tuple 
$(v, \lambda_v(u)) \in \Psi(u)$ satisfies 
$v \in [S_T \setminus (\OPT \setminus \set{o_m})] \cup \set{u_{T+1}}$ and 
$0 < \lambda_v(u) \le \min\set{c(u), c(v)}$. Moreover, we have:
\begin{align*}
\sum_{(v, \lambda_v(u)) \in \Psi(u)} \lambda_v(u) &= c(u), & \forall u \in O_{>T} \setminus \set{o_m}, \\
\sum_{\substack{u \in O_{>T} \setminus \set{o_m} \\ (v, \lambda_v(u)) \in \Psi(u)}} \lambda_v(u)
&\le c(v), & \forall v \in [S_T \setminus (\OPT \setminus \set{o_m})] \cup \set{u_{T+1}}.
\end{align*}
\end{lemma}

\begin{lemma}[Formal version of \Cref{lemma:gain_bound}]\label{lemma:gain_bound_restate}
 For each $u\in \cN$, let $X_u$ be an indicator for the event that $u\in O_{\le T}\setminus (S_T\cup\set{o_m})$ or $u\in S_T \setminus (\OPT\setminus\set{o_m})$. Then, conditioned on $\cE$,
 \begin{align*}
 f(S_T\cup\OPT) \le f(S_T\cup\set{o_m})+\sum_{u\in\cN} X_uf(u\mid S_u) +f(S^*)-f(S_T) + O\!\paren{\frac{k\Delta_f}{\eps_0}\log \frac{n}{\beta}},
 \end{align*}
where for any $j \in \set{1,\dots, \ell}$, $S_{u_j}$ denotes $S_{j-1}$, i.e., the set to which $u_j$ is considered for addition. For any $u \notin \{u_1,\dots,u_\ell\}$, define $S_u = \varnothing$.
\end{lemma}

\begin{proof}
By submodularity, we have
\begin{align}
 f(S_T\cup \OPT) - f(S_T\cup\set{o_m}) &\le \sum_{u\in \OPT\setminus(S_T\cup\set{o_m})} f(u\mid S_T) \nonumber\\
& \le \sum_{u\in O_{>T}\setminus\{o_m\}} f(u\mid S_T)
 + \sum_{u\in O_{\le T}\setminus(S_T\cup\set{o_m})} f(u\mid S_T) + (k-1)E . 
\label{eq:lemma_submodularity}
\end{align}
where the last inequality uses the fact that for any $v \in \OPT \setminus (O_{\le T} \cup O_{>T}\cup\set{o_m})$, it holds that $f(v \mid S_T) \le E$ by definition of $O_{>T}$. Since there are at most $k-1$ such elements, their total contribution is upper-bounded by $(k-1)E$.

Now recall that since $T\le r$, for every $j\in\set{1,\dots,T}$, by definition of $r$, the candidate set $C_{j-1}$ includes a feasible element $v$ such that $ f(v \mid S_{j-1})\ge E $.
Therefore, by \Cref{lemma:good_event} and under the conditioning on $\cE$,
\[
 f(u_{j}\mid S_{j-1}) \ge \max_{v\in \cC_{j-1}}
\left\{
\frac{f(v\mid S_{j-1})}{c(v)}
-
\frac{E}{c(v)}
\right\} \ge 0.
\]
Moreover, conditioned on $\cE$, since $S_T$ is among the candidates given to
\RNMT in Line~\ref{alg:sample_lastEM}, we have
\[
 f(S^*) \ge f(S_T) - \frac{4\Delta_f}{\eps}\log \frac{2k(n+1)}{\beta} \ge f(S_T) -E,
\]
where we have also used that $\eps_0\le \eps$. Therefore,
 \begin{align}
0&\le \sum_{u\in S_T\setminus(\OPT\setminus\{o_m\})}
 f(u\mid S_u)
 + f(S^*)-f(S_T)
 + E
 \label{eq:lemma_under_condition}.
\end{align}

We now consider two cases.

\textbf{Case 1: $T=r$.}
In this case, it must be that $O_{>T} = \varnothing$. Indeed, suppose toward a contradiction that $O_{>T} \neq \varnothing$. If $T = r = \ell$, then $c(S_\ell) \le c(\OPT \setminus \set{o_m})$. Since $O_{>T}\subseteq \OPT$, there exists an element in $O_{>T}$ that can be added to $S_\ell$ without violating the budget constraint, contradicting termination at iteration $\ell$. Thus, assume $T = r < \ell$. By the definition of $T$, any $v \in O_{>T}$ can be added to $S_T$ without violating the budget constraint. However, the definition of $r$ implies that for each such element, $f(v \mid S_T) \le E$, which contradicts the defining property of $O_{>T}$.

By \eqref{eq:lemma_submodularity} and \eqref{eq:lemma_under_condition}, we get
\begin{align}
f(S_T\cup \OPT) - f(S_T\cup\set{o_m})
&\le \sum_{u\in O_{\le T}\setminus(S_T\cup\set{o_m})} f(u\mid S_u)
 + \sum_{u\in S_T\setminus(\OPT\setminus\{o_m\})} f(u\mid S_u) \nonumber\\
&\qquad
 + f(S^*)-f(S_T)
 + kE
\end{align}
The inequality in the lemma follows.

\textbf{Case 2: $T<r$.} This case leverages the mapping defined in \Cref{lemma:mapping} to bound the term $ \sum_{u\in O_{>T}\setminus\{o_m\}} f(u\mid S_T)$. We have
\begin{align*}
\sum_{u \in O_{>T} \setminus \set{o_m}} f(u \mid S_T) 
&= \sum_{u \in O_{>T} \setminus \set{o_m}} \frac{f(u \mid S_T)}{c(u)} \cdot c(u) \\
&= \sum_{u \in O_{>T} \setminus \set{o_m}} \frac{f(u \mid S_T)}{c(u)} \sum_{(v, \lambda_v(u)) \in \Psi(u)} \lambda_v(u) \\
&\le \sum_{u \in O_{>T} \setminus \set{o_m}} \sum_{(v, \lambda_v(u)) \in \Psi(u)} \paren{\frac{f(v \mid S_v)}{c(v)} + \frac{E}{c(u)} }\cdot \lambda_v(u) \\
&= \sum_{u \in O_{>T} \setminus \set{o_m}} \Bigg[ E+ \sum_{(v, \lambda_v(u)) \in \Psi(u)} \paren{\frac{f(v \mid S_v)}{c(v)}}\cdot \lambda_v(u)\Bigg] \\
&\le \sum_{v \in (S_T \setminus (\OPT \setminus \{o_m\})) \cup \{u_{T+1}\}} \frac{f(v \mid S_v)}{c(v)} \cdot c(v) + kE \\
&\le \sum_{v \in S_T \setminus (\OPT \setminus \set{o_m})} f(v \mid S_v) +f(u_{T+1}\mid S_T) + \frac{8k\Delta_f}{\eps_0}\log\frac{2n^2}{\beta}.\\
&\le \sum_{v \in S_T \setminus (\OPT \setminus \set{o_m})} f(v \mid S_v) +f(S^*)-f(S_T) + \frac{4\Delta_f}{\eps}\log \frac{2k(n+1)}{\beta}+ \frac{8k\Delta_f}{\eps_0}\log\frac{2n^2}{\beta}.
\end{align*}

For the first inequality, note that no element $u \in O_{>T} \setminus \{o_m\}$ is selected before iteration $T$, and adding such an element to any $S_i$ for $i < T$ preserves feasibility. Consequently, $u$ appears in the candidate set provided to \GenRNMT at the iteration in which $v$ is selected. Therefore, by submodularity and the guarantee in \Cref{lemma:good_event}, we obtain
\[
\frac{f(u\mid S_T)}{c(u)} \le \frac{f(u\mid S_v)}{c(u)} \le \frac{f(v\mid S_v)}{c(v)} + \frac{E}{c(u)}.
\]

For the final inequality, note that $S_T \cup \{u_{T+1}\}$ is included among the candidates provided to \RNMT at Line~\ref{alg:sample_lastEM}. Under $\cE_2$, we have
\[
f(u_{T+1}\mid S_T) = f(S_T \cup \{u_{T+1}\}) - f(S_T) \le f(S^*) - f(S_T) + \frac{4\Delta_f}{\eps}\log \frac{2k(n+1)}{\beta}.
\]
Substituting this bound into \eqref{eq:lemma_submodularity} and using that $\eps_0 \le \eps$ completes the proof.

\end{proof}

The next lemma from \cite{han2021approximation} shows that $\sum_{u\in \mathcal{N}} X_u \cdot f(u \mid S_u)$ has the same expected value as $f(S_T) - f(\varnothing)$, where the expectation is over the Bernoulli random variables that determine if a selected element is discarded. The proof extends to our setting by conditioning on any arbitrary realization of the DP random variables. Crucially, because our algorithm samples all DP noise in advance, conditioning on these realizations does not introduce correlations with the Bernoulli trials. This allows us to obtain the equality in expectation conditioned on $\mathcal{E}$. We provide a proof sketch for completeness. Henceforth, all expectations are implicitly conditioned on $\mathcal{E}$.

\begin{lemma}[Lemma 3 in \citealp{han2021approximation}]\label[lemma]{lemma:expectation_equality}
 We have $\mathbb{E}[f(S_T) - f(\varnothing) ] = \mathbb{E}[\sum_{u\in \mathcal{N}} X_u \cdot f(u\mid S_u)]$, where $X_u$ is defined in \Cref{lemma:gain_bound}.
\end{lemma}

\begin{proof}[Proof Sketch]
Let us condition on a specific realization of the randomness of the DP mechanisms such that $\mathcal{E}$ occurs. By the law of total expectation, it suffices to show the equality in the lemma conditioned on any such realization.
For each $u\in \mathcal{N}$, define a random variable $R_u$ by $R_u = f(u\mid S_u)$ if $u\in S_T$ and $R_u=0$ otherwise. As $f(S_T) - f(\varnothing) = \sum_{u\in \mathcal{N}} R_u$, by linearity of expectation it suffices to show that $\mathbb{E}[R_u] = \mathbb{E}[X_u \cdot f(u\mid S_u)]$ holds for every $u\in\mathcal{N}$. Fix an element $u$ and condition further on all random choices up to the moment when $u$ is considered by the algorithm in Line~\ref{alg:sample_random_inclusion}, or on an execution where the algorithm terminates and $u$ is never considered. Under this conditioning, the set $S_u$ is fixed. If $u$ is never considered, or if $c(S_u\cup\{u\}) > c(\OPT\setminus\{o_m\})$ at the time it is considered, then $\tau(u) > T$, implying that both $R_u$ and $X_u$ are zero.
Otherwise, we have $\tau(u) \le T$. The algorithm adds $u$ into $S_u$ with probability $1/2$, and therefore $\mathbb{E}[R_u] = \frac{1}{2} f(u\mid S_u)$. We now consider the two cases in the definition of $X_u$. If $u \in \OPT \setminus \{o_m\}$, then $X_u=1$ if and only if $u$ is discarded by the algorithm. If $u \notin \OPT \setminus \{o_m\}$, then $X_u=1$ if and only if $u$ is selected by the algorithm. In both cases, $X_u=1$ with probability $1/2$, and hence $\mathbb{E}[X_u \cdot f(u\mid S_u)] = \frac{1}{2} f(u\mid S_u)$. Since the conditional expectations agree in all cases, the claim follows. As the equality holds for any realization, it also holds in expectation conditioned on $\mathcal{E}$.
\end{proof}

The proof of the utility guarantee is completed by utilizing the following lemma.
\begin{lemma}[\citet{buchbinder2014submodular}]\label[lemma]{lemma:stochastic_bound}
Let $g : 2^{\mathcal{N}} \to \mathbb{R}$ be a non-negative submodular function. Denote by $Y$ a random subset of $\mathcal{N}$ where each element appears with probability at most $p$ (not necessarily independently). Then, $\mathbb{E}[g(Y)] \ge (1-p) g(\varnothing)$.
\end{lemma}

\begin{lemma}[\Cref{lemma:sample_utility_proof_main} restated]
\Cref{algorithm:subsample_greedy} outputs a feasible set $S^*$ such that, with probability $1-\beta$,
\[
 \mathbb{E}[f(S^*)] \ge \frac{1}{4} f(\OPT) - O\Big(\frac{k\Delta_f}{\eps_0}\log \frac{n}{\beta}\Big),
\]
\end{lemma}

\begin{proof}
The algorithm outputs a feasible set since all partial solutions considered are feasible, and only feasible sets are included in $\mathcal{S}$ as defined in Line~\ref{alg:sample_defineS}. By \Cref{lemma:small_noise_event}, we have $\Pr[\mathcal{E}] \ge 1-\beta$. For the remainder of the proof, we condition on $\mathcal{E}$, and all expectations are implicitly conditioned on this event.
Combining \Cref{lemma:gain_bound_restate} and \Cref{lemma:expectation_equality}, we have:
\begin{align*}
 \mathbb{E}[f(S_T \cup \OPT)] 
 &\le \mathbb{E}[f(S_T \cup \{o_m\})] + \mathbb{E}[f(S_T) - f(\varnothing) ] + \mathbb{E}[f(S^*) -f(S_T)] + O\Big(\frac{k\Delta_f}{\eps_0}\log \frac{n}{\beta}\Big) \\
 &\le 2 \mathbb{E}[f(S^*)] + O\Big(\frac{k\Delta_f}{\eps_0}\log \frac{n}{\beta}\Big).
\end{align*}
The first inequality is by \Cref{lemma:gain_bound_restate} and \Cref{lemma:expectation_equality}.
The second inequality follows from non-negativity and the fact $S_T \cup \{o_m\}$ is among the candidates provided to \RNMT in Line~\ref{alg:sample_lastEM}. Therefore, conditioned on $\cE$, we have $f(S^*) \ge f(S_T \cup \{o_m\}) - O(\frac{\Delta_f}{\eps} \log \frac{n}{\beta})$.

Each element of $\mathcal{N}$ is included in $S_T$ with probability at most $1/2$ due to the Bernoulli sampling in Line~\ref{alg:sample_random_inclusion}. Hence, applying \Cref{lemma:stochastic_bound} to the submodular function $g(S) \eqdef f(S \cup \OPT)$, we get
\[
 \mathbb{E}[f(S_T \cup \OPT)] = \mathbb{E}[g(S_T)] \ge \frac{1}{2} g(\varnothing) = \frac{1}{2} f(\OPT).
\]
Combining these inequalities, we obtain $\frac{1}{2} f(\OPT) \le 2 \mathbb{E}[f(S^*)] + O(\frac{k\Delta_f}{\eps_0}\log \frac{n}{\beta})$. Rearranging completes the proof.
\end{proof}
\section{Additional Empirical Results}
\label{appendix:experiments}

\begin{figure*}
    \centering
    \begin{subfigure}[t]{0.28\linewidth}
        \centering
        \includegraphics[width=\linewidth, height=0.12\textheight, keepaspectratio]{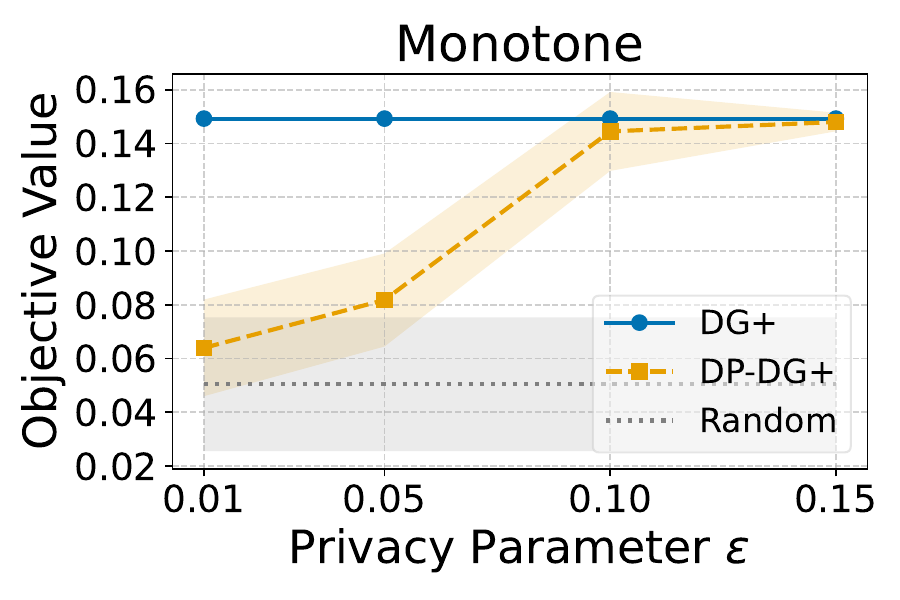}
        \caption{Objective value for varying $\eps$.}
        \label{fig:mon_val_for_eps}
    \end{subfigure}
    \hspace{8mm} 
    \begin{subfigure}[t]{0.28\linewidth}
        \centering
        \includegraphics[width=\linewidth, height=0.12\textheight, keepaspectratio]{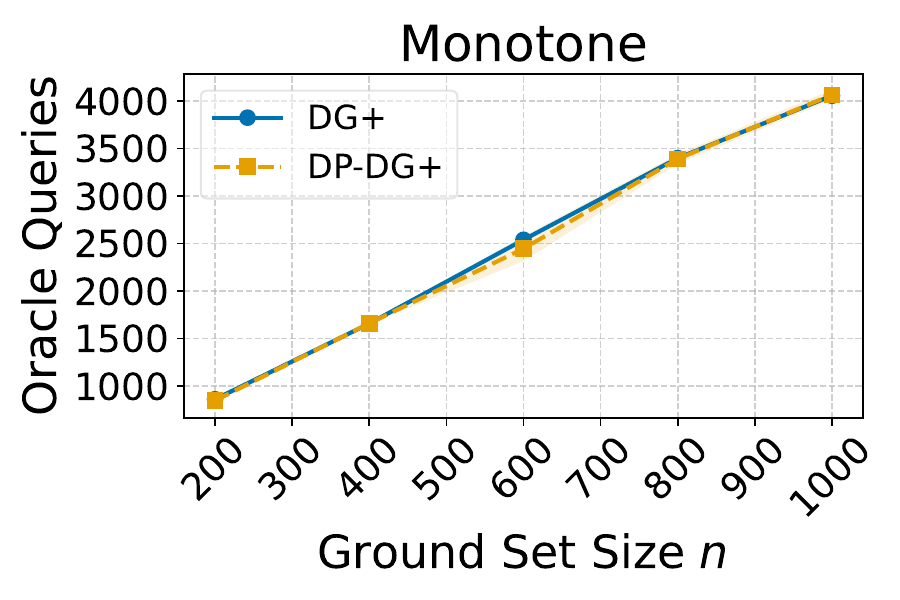}
        \caption{Number of queries for varying $n$.}
        \label{fig:mon_queries_for_n}
    \end{subfigure}
    \caption{Performance of \DPDGP at an increased scale.}
    \label{fig:scale}
\end{figure*}

\begin{figure*}
    \centering
    \begin{subfigure}[t]{0.24\linewidth}
        \centering
        \includegraphics[width=\linewidth,height=0.12\textheight, keepaspectratio]{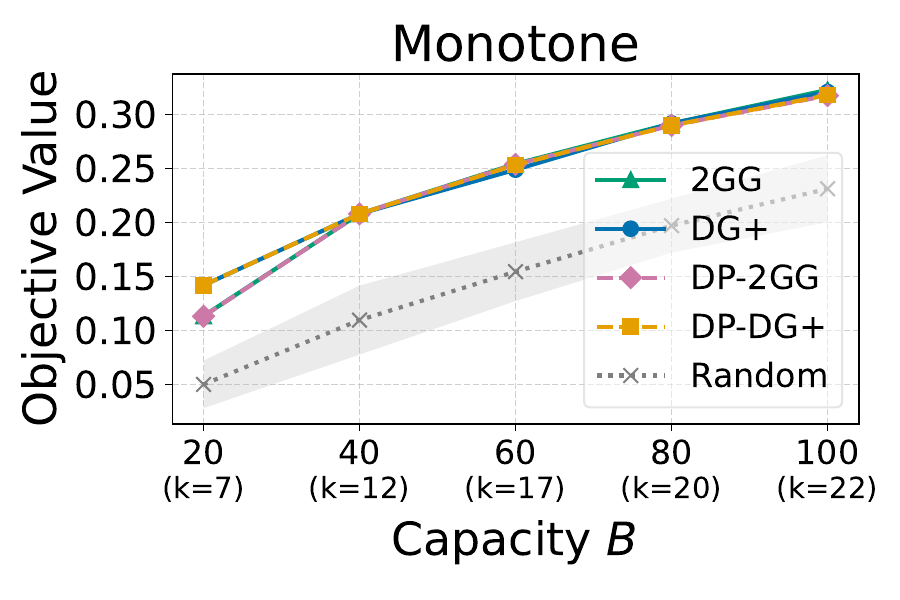}
        \caption{$n=30$.}
        \label{fig:mon_val_for_B}
    \end{subfigure}
    \hspace{8mm} 
    \begin{subfigure}[t]{0.24\linewidth}
        \centering
        \includegraphics[width=\linewidth, height=0.12\textheight, keepaspectratio]{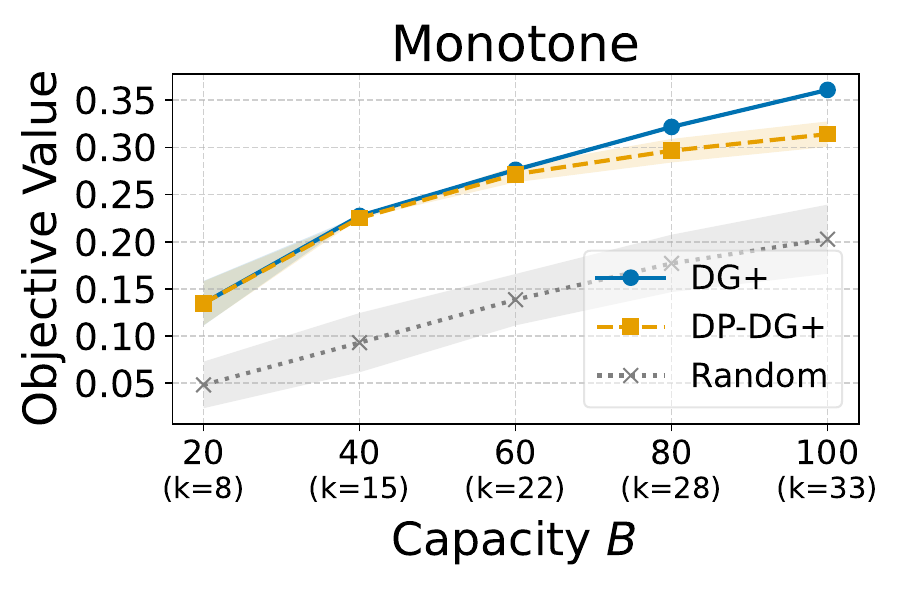}
        \caption{$n=100$.}
        \label{fig:fast_mon_val_for_cap}
    \end{subfigure}
    \hspace{8mm} 
    \begin{subfigure}[t]{0.24\linewidth}
        \centering
        \includegraphics[width=\linewidth, height=0.12\textheight, keepaspectratio]{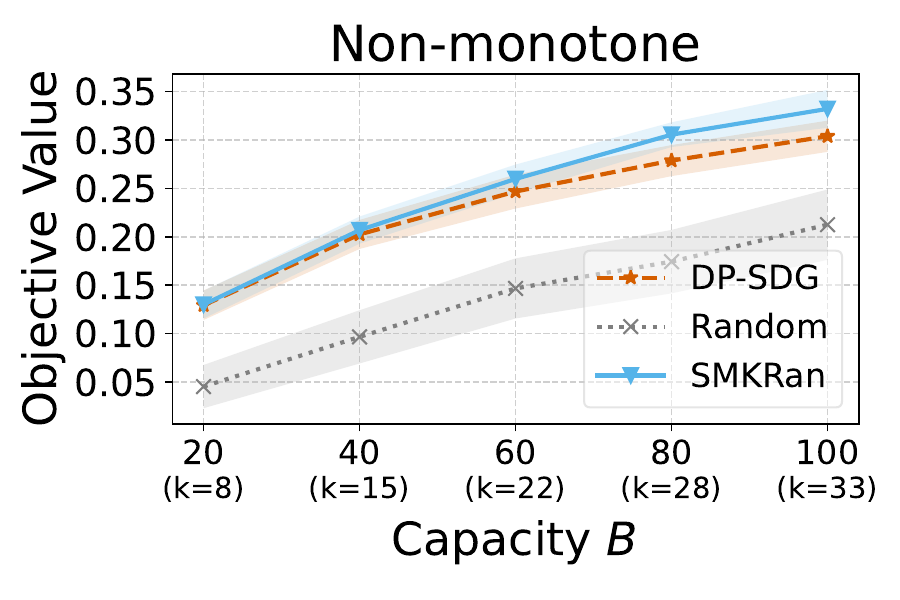}
        \caption{$n=100$.}
        \label{fig:nonmon_val_for_cap}
    \end{subfigure}
    \caption{Objective value for varying budget $B$. The induced maximal cardinality $k$ of a feasible solution is shown below the $x$-axis.}
    \label{fig:B_impact}
\end{figure*}

In this section, we present additional empirical results. We evaluate \DPDGP with a default setting of $n=100, B=20$ ($k=8$), and scale to larger values of $n$, mirroring our analysis of \DPSDG in \Cref{sec:experiments}. We also study the impact of the knapsack capacity parameter $B$ on the utility of all algorithms.

\paratitle{Higher-Scale Evaluations for \DPDGP}
The results in \Cref{fig:mon_val_for_eps} demonstrate that \DPDGP maintains strong utility relative to the non-private \DGP baseline, achieving nearly identical performance even under a strict setting of $\eps=0.1$, where it remains within 3\% of the baseline, while \Random is 67\% lower. Furthermore, \Cref{fig:mon_queries_for_n} highlights the algorithm's scalability: the number of queries grows linearly with $n$, consistent with the theoretical bound.

\paratitle{Impact of $B$ on Utility}
\Cref{fig:B_impact} shows the objective value obtained under different settings for varying values of $B$. As predicted by our theoretical error bounds, increasing $B$ leads to a decrease in the relative utility of the DP algorithms, with this degradation becoming most pronounced for the larger ground set size $n=100$. Nevertheless, our proposed algorithms consistently maintain competitive utility across all regimes.
As shown in \Cref{fig:fast_mon_val_for_cap} for \DPDGP and in \Cref{fig:nonmon_val_for_cap} for \DPSDG, at $n=100$ and $B=80$ (corresponding to $k=28$), both algorithms remain within 8\% of their respective non-private counterparts, \DGP and \SmkRan, while \Random performs substantially worse, achieving objective values at least 42\% below the non-private baselines. Furthermore, \Cref{fig:mon_val_for_B} shows that, in the monotone experiments with $n=30$ (corresponding to $k=22$), \DPTwoGG and \DPDGP remain within 1\% of the non-private \TwoGG.

\end{document}